\documentclass[11pt,twoside]{article}
\usepackage{fancyhdr}
\usepackage[colorlinks,citecolor=blue,urlcolor=blue,linkcolor=blue,bookmarks=false]{hyperref}
\usepackage{amsfonts,epsfig,graphicx}
\usepackage{afterpage}
\usepackage{comment}
\usepackage{amsmath,amssymb,amsthm} 
\usepackage{fullpage}
\usepackage[T1]{fontenc} 
\usepackage{epsf} 
\usepackage{graphics} 
\usepackage{amsfonts,amsmath}
\usepackage[sort]{natbib} 
\usepackage{psfrag,xspace}
\usepackage{color,etoolbox}

\def\inprob{\stackrel{p}{\rightarrow}}

\def\ind{\perp\!\!\!\perp}
\def\T{{ \mathrm{\scriptscriptstyle T} }}

\newcommand{\var}{\text{var}}
\newcommand{\cov}{\text{cov}}

\newcommand{\Pb}{\mathbb{P}}
\newcommand{\Pn}{\mathbb{P}_n}

\newcommand{\E}{\mathbb{E}}
\newcommand{\R}{\mathbb{R}}

\DeclareMathOperator*{\argmin}{arg\,min}

\DeclareMathOperator{\sgn}{sgn}
\DeclareSymbolFont{bbold}{U}{bbold}{m}{n}
\DeclareSymbolFontAlphabet{\mathbbold}{bbold}
\newcommand{\one}{\mathbbold{1}}
\newtheorem{theorem}{Theorem}
\newtheorem{lemma}{Lemma}
\newtheorem{corollary}{Corollary}
\newtheorem{proposition}{Proposition}

\theoremstyle{definition}

\newtheorem{example}{Example}
\theoremstyle{remark}
\newtheorem{assumption}{Assumption}

\newtheorem{remark}{Remark}

\usepackage{chngcntr}
\newcounter{preexample}
\counterwithin{example}{preexample}

\newcommand{\examplegroup}{\refstepcounter{preexample}}

\let\hat\widehat

\begin{document}

\def\spacingset#1{\renewcommand{\baselinestretch}%
{#1}\small\normalsize} \spacingset{1}

\raggedbottom
\allowdisplaybreaks[1]


  \title{\vspace*{-.45in} {Semiparametric Counterfactual Density Estimation}}
  \author{\\ Edward H. Kennedy, Sivaraman Balakrishnan, Larry Wasserman  \\ \\
    Department of Statistics \& Data Science \\
    Carnegie Mellon University \\ \\ 
    \texttt{\{edward, siva, larry\} @ stat.cmu.edu} \\
\date{}
    }

  \maketitle
  \thispagestyle{empty}

\begin{abstract}
Causal effects are often characterized with averages, which can give an incomplete picture of the underlying counterfactual distributions. Here we consider estimating the entire counterfactual density and generic functionals thereof.  We focus on two kinds of target parameters.  The first is a density approximation, defined by a projection onto a finite-dimensional model using a generalized distance metric, which includes $f$-divergences as well as $L_p$ norms. The second is the distance between counterfactual densities, which can be used as a more nuanced effect measure than the mean difference, and as a tool for model selection. We study nonparametric efficiency bounds for these targets, giving results for smooth but otherwise generic models and distances. Importantly, we show how these bounds connect to means of particular non-trivial functions of counterfactuals, linking the problems of density and mean estimation. We go on to propose doubly robust-style estimators for the density approximations and distances, and study their rates of convergence, showing they can be optimally efficient in large nonparametric models. We also give analogous methods for model selection and aggregation, when many models may be available and of interest.  Our results all hold for generic models and distances, but throughout we highlight what happens for particular choices, such as $L_2$ projections on linear models, and KL projections on exponential families. Finally we illustrate by estimating the density of CD4 count among patients with HIV, had all been treated with combination therapy versus zidovudine alone, as well as a density effect.  
Our results suggest combination therapy may have increased CD4 count most for high-risk patients.
Our methods are implemented in the freely available R package \emph{npcausal} on GitHub. 
\end{abstract}

\medskip


\section{Introduction}

It is very common in causal inference to quantify causal effects with means. The classic average treatment effect (ATE) parameter, for instance, measures the difference in mean outcome had all versus none in a population been treated. This can certainly be a useful summary, but it can also miss potentially important differences in the distributions of the counterfactual outcomes, beyond a simple mean shift. To illustrate, consider the densities in Figure \ref{fig:densities}, which all have exactly the same mean and variance. These would be indistinguishable with the ATE,  or any other measure that did not look past the first two moments. \\

\begin{figure}[h!]
\begin{center}
{\includegraphics[width=.67\textwidth]{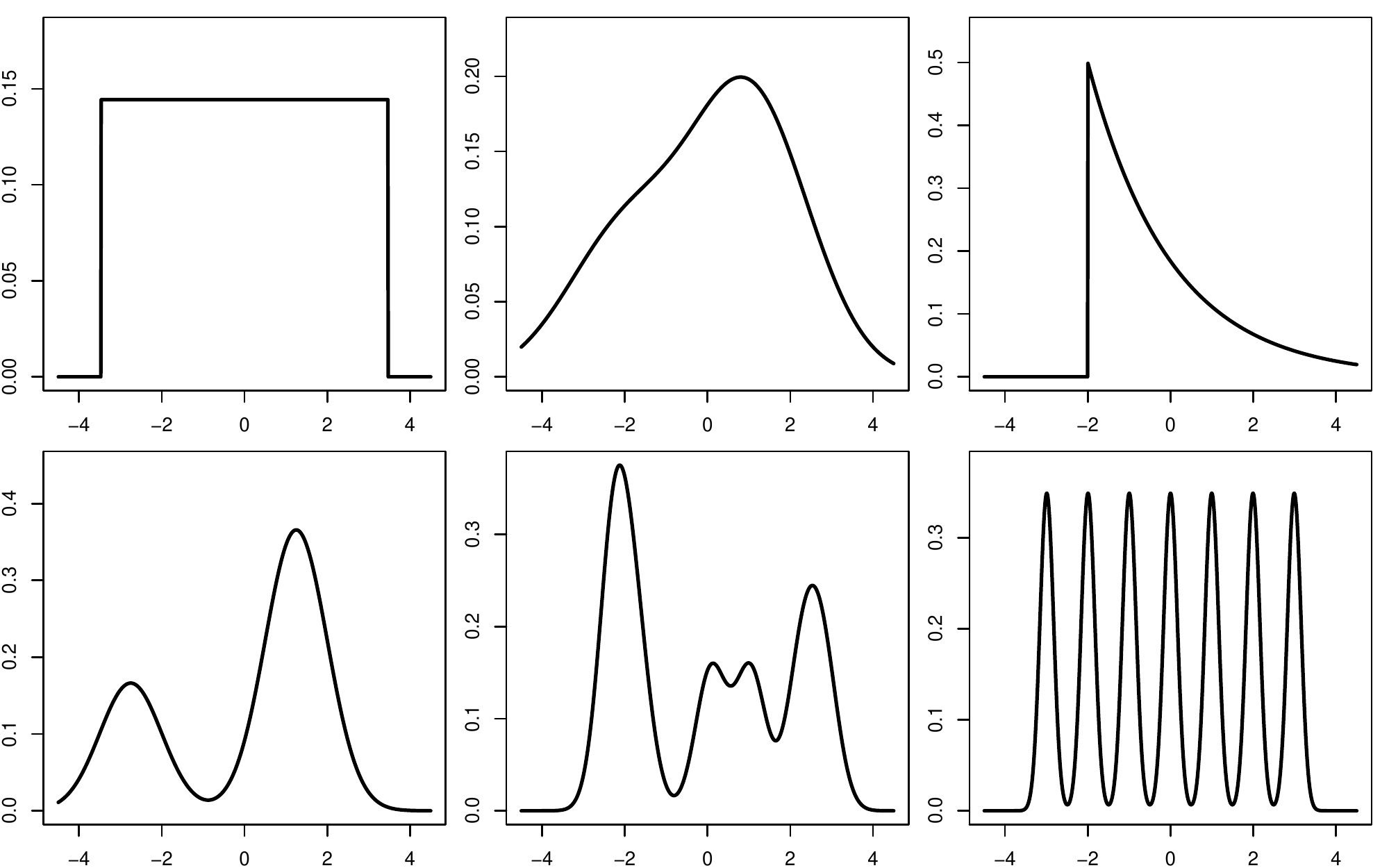}}
\caption{\em Densities for six distributions, all with mean zero and variance four.} \label{fig:densities}
\end{center}
\end{figure}

In general, it can be very practically useful to know the shape of the counterfactual density. If the counterfactual density differed at all under treatment versus control (or under any other generic interventions), this would imply treatment had \emph{some} effect, even if the ATE were zero. The presence of skew would indicate that some subjects have relatively extreme responses to treatment; next steps could include trying to understand who these subjects are, and why their responses are unusual. Similarly, multimodal structure could point to  the existence of underlying subgroups with differential responses to treatment, which could be important for optimizing treatment policies. Contrasting the shape of the density under different interventions could inform hypotheses about how treatment works, e.g., perhaps it works by reducing variance, or driving up negative outcomes. This could help enhance future versions of treatment, or motivate the development of new treatments altogether.   \\

There is a large literature on distributional treatment effects defined in terms of quantiles, or cumulative distribution functions (CDFs), with a similar goal of moving beyond simple mean summaries to study the entire counterfactual distribution \citep{abadie2002bootstrap, melly2005decomposition, chernozhukov2005iv, machado2005counterfactual, firpo2007efficient, rothe2010nonparametric, wang2010empirical, fortin2011decomposition, zhang2012causal, chernozhukov2013inference, frolich2013unconditional, diaz2017efficient, wang2018quantile}. However, the challenges and methods are substantially different for density estimation. This is largely a result of the fact that the CDF at $y$ is the mean of the thresholded outcome $\one(Y \leq y)$, so that counterfactual CDF estimation mostly reduces to counterfactual mean estimation, after replacing the outcome with an indicator. A related difference is that the CDF is pathwise differentiable in a nonparametric model, whereas the density function is not \citep{bickel1993efficient,van2003unified}. This is also true in the standard observational setup, where CDFs can be estimated at $n^{-1/2}$ rates with sample averages, while density estimation requires more careful balancing of bias and variance, with slower rates arising depending on underlying smoothness \citep{wasserman2006all, tsybakov2009introduction}. Beyond this issue of statistical complexity, there are other trade-offs in targeting CDFs versus densities. One is that, although CDFs are easier to estimate nonparametrically, densities are arguably more visually appealing and interpretable to practitioners. We view CDFs and densities as complementary pieces of the distributional puzzle. \\

Unlike distribution function estimation, the literature on counterfactual density estimation appears much more sparse. In what was perhaps the first study of the problem, \citet{dinardo1996labor} used a reweighted kernel estimator to estimate effects of US labor market factors on wages.  However, the statistical properties of this proposed approach were not examined. 
\citet{robins2001inference} proposed a doubly robust version of the reweighted kernel estimator, and conjectured it would achieve usual density estimation rates under smoothness and other conditions. \citet{van2003cross} and \citet{rubin2006extending} studied general cross-validation-based approaches for model selection in the presence of nuisance functions, and suggested minimizing  counterfactual KL or $L_2$ loss for density estimation, but did not detail the statistical properties. More recently, \citet{westling2020unified} tackled the related problem of density estimation for right-censored outcomes, 
proposing new estimators that can attain $n^{-1/3}$ rates, but under an assumption that the density is monotone. \citet{kim202xcausal} analyzed  a version of the doubly robust estimator from \citet{robins2001inference}, showing its conjectured oracle properties, and used it to estimate the (nonsmooth) $L_1$ distance between counterfactual densities.  \\

Somewhat surprisingly, none of the known work above on counterfactual density estimation considers a semiparametric approach, where the density is approximated with a finite-dimensional model. Our work aims to fill this gap in the literature, while also providing data-driven model selection and aggregation tools. A separate contribution is our study of generic density-based effects, which characterize the distance between counterfactual densities, using a generalized notion of distance that includes $f$-divergences as well as $L_p$ norms. \\

The structure of our paper is as follows. After introducing some basics and causal assumptions in Section \ref{sec:setup}, in Section \ref{sec:targetparams} we detail the different kinds of target parameters we consider. The first (described in Section \ref{sec:densityfns}) is an approximation of the density itself, defined by a projection onto a finite-dimensional model (\ref{sec:densityfns_models}) using a generalized distance metric (\ref{sec:densityfns_distances}), which includes $f$-divergences as well as $L_p$ norms. Importantly, we show in Section \ref{sec:densityfns_momcon} that projection parameters for smooth models and distances can be framed as solutions to moment conditions, providing a link between counterfactual densities and means (of functions of counterfactuals). The second parameter we consider (described in Section \ref{sec:denseff}) is the distance between counterfactual densities, which can be used as a new more nuanced effect measure, or as a tool for model selection (as in Section \ref{sec:modsel}). In Section \ref{sec:efficiency} we study nonparametric efficiency bounds, by characterizing the efficient influence functions of  approximated density functions in Section \ref{sec:densityfnseif}, and density effects in Section \ref{sec:denseffeif}. These follow from a master lemma in Section \ref{sec:efficiency}, which gives a von Mises expansion for generic integral functionals of the counterfactual density, and so may be of independent interest. In Section \ref{sec:estimation} we propose doubly robust-style estimators for the density approximations and distances, and study their rates of convergence, showing for example that they can be $n^{-1/2}$ consistent, asymptotically normal, and optimally efficient under weak high-level conditions on nuisance estimation error. All our results hold for smooth but otherwise arbitrary models and distances. However, in various corollaries, we also highlight specific expressions for typical choices of models and distances, such as $L_2$ projections on linear models, and KL projections on exponential families.   Finally in Section \ref{sec:illustration} we use our proposed methods to estimate counterfactual densities and density effects of combination therapy (versus zidovudine alone) on CD4 count, among patients with HIV. Our results show treatment effects beyond a mean shift, suggesting that combination therapy may have increased CD4 count most for high-risk patients. \\

\section{Setup}
\label{sec:setup}

We assume access to  an iid sample $(Z_1,...,Z_n)$ of $Z=(X,A,Y) \sim \Pb$
where $X \in \R^d$ are covariates, $A \in \R$ is a treatment or exposure, and $Y \in \R$ is a continuous outcome. We let 
\begin{align}
\pi_a(x) &= \Pb(A=a\mid X=x) \label{eq:ps} \\
\int_{\mathcal{B}} \eta_a(y \mid x) \ dy &= \Pb(Y \in \mathcal{B} \mid X=x,A=a) \ \text{ for measurable $\mathcal{B}$}
\end{align}
denote the propensity score (i.e., chance of being treated at level
$A=a$ given covariates) and conditional outcome density,
respectively. In this work we focus on discrete treatments, but in a companion paper we consider the continuous case.  \\

We study ``semiparametric'' estimation of the covariate-adjusted marginal density 
\begin{equation}
p_a(y) = \int \eta_a(y \mid x) \ d\Pb(x) \label{eq:dens}
\end{equation}
i.e., the conditional outcome density averaged over the covariates, as well as functionals thereof. \\

\begin{remark} \label{rem:nonpar}
Although we refer to our work in this paper as semiparametric, in
reality it is all done within a fully nonparametric model. As
described in more detail starting in Section
\ref{sec:densityfns}, the models we consider are only ever used as
tools for defining nonparametric approximations, and corresponding
projection parameters, and are never assumed to be correct
descriptions of the underlying true data-generating process. Further,
our results on estimating counterfactual density \emph{functionals}
(e.g., Sections \ref{sec:denseff} and \ref{sec:denseffeif}) do not
require any approximating models, and so are nonparametric in the
usual sense. \\
\end{remark}

\bigskip

We note that the density \eqref{eq:dens} is
different from the marginal density of $Y$, since the treatment is fixed at $A=a$
in the conditioning; it is also not equal to the unadjusted conditional density $p(y \mid a)$. 
Instead, \eqref{eq:dens} is the density $p(y^a)$ of the counterfactual variable $Y^a$ (i.e., 
the outcome that \emph{would have been observed} if treatment were set to $A=a$),
 if the following assumptions hold:
\begin{assumption}[Positivity]
$\Pb\{ \pi_a(X) \geq \epsilon \}=1$ for some $\epsilon>0$. 
\end{assumption}
\begin{assumption}[Consistency]
$Y=Y^a$ if $A=a$.
\end{assumption}
\begin{assumption}[Exchangeability]
$A \ind Y^a \mid X$. 
\end{assumption} 
Positivity ensures all subjects have some chance at receiving
treatment level $A=a$. Consistency can be viewed as ruling out
interference, for example, where a subject's counterfactual can depend not only on
how they were treated, but how other subjects were treated as
well. Exchangeability says the treatment is as good as randomized
within levels of the observed covariates, and requires that
sufficiently many relevant confounders are collected. Each of these
assumptions can be weakened in various ways, at the expense of losing
point identification of the marginal counterfactual distribution.
  Nonetheless, under only the positivity
assumption, all our statistical results will hold relative to the
observational quantity in \eqref{eq:dens}, regardless of whether the
causal Assumptions 2--3 are violated or not. \\

\section{Target Parameters}
\label{sec:targetparams}

In this section we detail the two kinds of quantities we consider
estimating. The first is an approximation of the counterfactual
density itself, defined via a projection in some distributional distance. The
second is a distance measure, e.g., a density-based causal effect measuring the difference
between counterfactual densities in terms of general
$f$- or other divergences. The latter gives a more nuanced picture of how the
counterfactual densities differ, compared to the usual ATE, for example. 
Finally in Section \ref{sec:modsel} we describe how these two kinds of target quantities 
can be adapted for the purposes of model selection and aggregation.  

\subsection{Density Functions}
\label{sec:densityfns}

\subsubsection{Models}
\label{sec:densityfns_models}

First we consider approximations of the counterfactual density $p_a(y)$ based on some
specified model $\{g(y;\beta) : \beta \in \R^d \}$. 
We mostly focus on the finite-dimensional parametric case with $\beta \in \R^d$, 
but more generally one could take $\beta$ to be infinite-dimensional in some $L^p$ space, 
or to belong to a subset of $\R^d$ such as the standard simplex. 
Note that $\beta(a)$ depends on $a$
but for now we suppress this dependence in the notation
and simply write $\beta$.
Here are some examples. \\

\examplegroup

\begin{example}[Exponential family] \label{ex:expfam}
Let $b(y)=\{b_1(y),...,b_d(y)\}^\T$ denote a vector of known basis functions. Then we can project onto the exponential family
\begin{equation}
g(y;\beta) =  \exp\Big\{ \beta^\T b(y) - C(\beta) \Big\} \label{eq:expfam}
\end{equation}
where $C(\beta) = \log \int \exp\{ \beta^\T b(y) \} \ dy$ so that
$\int g(y;\beta) \ dy = 1$. Typical exponential family notation takes
$\beta_1=1$ and sets $b_1(y)=\log h(y)$ for some known base
measure $h$. 
\end{example}

Although we refer to Example \ref{ex:expfam}  as an exponential family, it can just as well be viewed as a truncated series expansion used together with a log link function. In the next example we consider a truncated series with an identity link function. 

\begin{example}[Truncated series] \label{ex:series}
Let $b(y)=\{b_1(y),...,b_d(y)\}^\T$ denote a vector of known basis functions, and $q(y)$ a known base density (e.g., uniform). Then we can project onto the linear basis expansion
\begin{equation*}
g(y;\beta) = q(y) + \sum_{j=1}^d \beta_j b_j(y) 
\end{equation*}
where we can take $\int b_j(y) \ dy = 0$ so that the projection integrates
to one.  A natural choice when $Y \in [0,1]$ would be to take $q(y)=1$ and $b(y)$ the cosine basis
\begin{equation}
b_j(y) = \sqrt{2} \cos(\pi j y) \label{eq:cosbasis}
\end{equation}
which satisfies $\int b_j(y) \ dy=0$ and $\int b_j(y) b_k(y) \ dy = \one(j=k)$ on the unit interval. 
One could alternatively take $q(y)=0$ and let $b_j(y)$ be (the linear span of) a collection of $d$
candidate densities, in which case the above could be viewed as a linear
aggregation \citep{rigollet2007linear}.
Another related option would be to use a linear approximation for the square root
of the density $\sqrt{g(y;\beta)}=\sum_j\beta_j b_j(y)$, so that
$g(y;\beta) = \sum_j \sum_k \beta_j \beta_k b_j(y) b_k(y)$
\citep{pinheiro1997estimating, chen2002monte}. Then the model would
integrate to one if the basis functions were orthonormal ($\int b_j(y)
b_k(y) \ dy=0$ and $\int b_j(y)^2 \ dy = 1$) and $\sum_j \beta_j^2=1$. 
\end{example}

\begin{example}[Gaussian mixture model]
Let $(\mu_1,...,\mu_k)$ denote a vector of means, $(\sigma_1,...,\sigma_k)$ a vector of positive standard deviations, $(\varpi_1,...,\varpi_k)$ positive mixing proportions with $\sum_j \varpi_j =1$, and $\phi$ the standard normal density. Then the standard Gaussian mixture model is
\begin{equation*}
g(y;\beta) = \sum_{j=1}^k  \varpi_j \left( \frac{1}{\sigma_j} \right) \phi \left( \frac{y-\mu_j}{\sigma_j} \right)
\end{equation*}
where $\beta = \{(\varpi_1,\mu_1,\sigma_1^2), ...., (\varpi_k,\mu_k,\sigma_k^2)\}$. \\
\end{example}

\bigskip

Now, based on the above approximations, a primary goal is to estimate the projection parameter
\begin{equation}
\beta_0  = \argmin_{\beta \in \R^p} \ D_f\Big(p_a(y) , g(y; \beta)\Big)  \label{eq:target}
\end{equation}
where $D_f$ is a distributional distance measure of the form
\begin{equation}
D_f(p,q)  = \int f( p , q ) q(y) \ dy \label{eq:distance}
\end{equation}
for some given discrepancy function $f: \R^2 \rightarrow \R$.  \\

\begin{remark}
In contrast to typical $f$-divergences \citep{renyi1961measures,
sason2016f,ali1966general,csiszar1967information}, we allow the
function $f$ to have two arguments, one for each distribution; this
allows us to capture not only $f$-divergences but also other distances
such as those based on $L_p$ norms. (The usual $f$-divergence takes
$f(p,q)=f(p/q)$ for some single-argument function, and so only depends
on the density ratio). In a slight abuse of terminology we sometimes refer to
 \eqref{eq:distance} as  a distance, even though for some of our choices of $f$ it will be an asymmetric divergence not satisfying the triangle inequality. \\
\end{remark}

Before giving examples of distances, we first discuss the
interpretation of our projection parameter \eqref{eq:target}. Mathematically, $\beta_0$ is the parameter of the best-fitting model of the
form $g(y;\beta)$, i.e., the parameter value that makes $g(y;\beta)$
closest (in corresponding distance) to the true density $p_a(y)$. If
the model $g$ is correctly specified, then $D_f( p_a(y) , g(y;\beta_0))=0$ and so 
$g(y;\beta_0)=p_a(y)$ is simply the true counterfactual density; however, the projection
\eqref{eq:target} remains well-defined even under model
misspecification. This is akin to the well-known concept of a \emph{best
linear predictor} in standard linear regression \citep{white1980using}. The
projection approach, where a model is not assumed correct but instead
only used for defining approximations, has been used widely throughout
statistics \citep{huber1967behavior, beran1977minimum, buja2019models1, buja2019models2,
white1982maximum, white1996estimation, tsybakov2003optimal,
wasserman2006all, rinaldo2010generalized, rakhlin2017empirical} 
as
well as in causal inference \citep{van2006statistical, neugebauer2007nonparametric,
chernozhukov2018generic, kennedy2019robust, cuellar2020non,
semenova2020debiased}, though not in the counterfactual density estimation
context. \\

\begin{remark} \label{rem:models}
Since we only use models as tools to define
approximations, all our results are formally nonparametric, as
mentioned in Remark \ref{rem:nonpar} and illustrated in subsequent
theorems. This raises some interesting philosophical issues about the
role of assumptions and corresponding bias-variance trade-offs. In
particular, we can imagine a rough taxonomy of stances one might take
in estimation problems like this one:
\begin{enumerate}
 \setlength\itemsep{0em}
\item[(i)] model-ist: My finite-dimensional/parametric representation is \emph{the} correct one.
\item[(ii)] model-agnostic: I may \emph{use} a finite-dimensional model, but I do not know or require that it is a perfectly accurate picture of the truth.
\item[(iii)] anti-model-ist: No parametric model I can imagine contains the truth, and I do not care about approximations.
\end{enumerate}
The model agnostic view
is often captured
by the famous quotes
``All models are wrong but some are useful'' (George Box)
and 
``Use models but don't believe them''
(possibly due to John Tukey).
Of course, in practice, how much one relies on models is a continuum,
and so any particular approach may not fall entirely in one of the
three camps above. Similarly, our taxonomy uses parametric models as a
benchmark, but one could just as well replace with a different
assumption set (e.g., H{\"o}lder-smooth with index $s\geq 4$ versus $s
< 4$). Nevertheless we find the above framing useful if imperfect.
In this paper, we mostly take the stance of the model-agnostic, though
we flirt with anti-model-ism in the data-driven model selection
approaches of Sections \ref{sec:modsel} and \ref{sec:modselest} (and
we are fully anti-model-ist in a companion paper). We also accept that
each approach has advantages and disadvantages. The model-ist will do
 well when the model is correct, but could unknowingly suffer
large bias otherwise. The anti-model-ist is most free from the
constraints of human imagination (as they do not need to posit a parametric model), but with a more ambitious target can
also suffer larger errors. The model-agnostic has a bit of
the best of both worlds: when the model is correct, they may hope to do nearly as
well as the model-ist, and when the model is wrong, their inference can 
still be valid for a still well-defined approximation. Of course, if the
model is \emph{very} wrong, the approximation may not be practically
useful, no matter how well-defined it is; thus there can be important challenges in defining a useful approximating model and distance.  \\
\end{remark}

\subsubsection{Distances}
\label{sec:densityfns_distances}

Now we give some examples of the distances we focus on in this paper: \\

\examplegroup

\begin{example}[$L_2^2$]
If $f(p,q)=(p-q)^2/q$ then $D_f(p,q)=\|p-q\|_2^2$ is the squared $L_2$ distance
\begin{equation*}
\|p_a(y) - g(y;\beta) \|_2^2 = \int  \Big( p_a(y) - g(y;\beta) \Big)^2 \ dy . 
\end{equation*}
\end{example}

\begin{example}[Kullback-Leibler]
If $f(p,q)=(p/q) \log(p/q)$ then $D_f(p , q)=\text{KL}(p,q)$ is the Kullback-Leibler divergence 
\begin{equation*}
\text{KL}\Big(p_a(y) , g(y;\beta)\Big) = \int  \log\left( \frac{p_a(y)}{g(y;\beta)} \right)  p_a(y) \ dy . 
\end{equation*}
\end{example}

\begin{example}[$\chi^2$]
If $f(p,q)=(p/q-1)^2$ then $D_f(p , q)=\chi^2(p,q)$ is the $\chi^2$ divergence 
\begin{equation*}
\chi^2\Big(p_a(y) , g(y;\beta)\Big) = \int \frac{ \{ p_a(y) - g(y;\beta) \}^2 } {g(y;\beta)}  \ dy  . 
\end{equation*}
\end{example}

\begin{example}[Hellinger]
If $f(p,q)=(\sqrt{p/q}-1)^2$ then $D_f(p , q)=H^2(p,q)$ is the squared Hellinger divergence 
\begin{equation*}
H^2\Big(p_a(y) , g(y;\beta)\Big) = \int \left( \sqrt{p_a(y)} - \sqrt{g(y;\beta)}  \right)^2 \ dy . 
\end{equation*}
\end{example}

\begin{example}[Smoothed Total Variation]
If $f(p,q) = \frac{1}{2q} | p-q|= \frac{(p-q)}{2q} \sgn(p-q)$ then $D_f(p , q) = \text{TV}(p,q)=\frac{1}{2} \| p - q \|_1 $ is the total variation distance (and half the $L_1$ distance). Note $f(p,q)$ is not differentiable at $p/q=1$. Smooth versions can be obtained by approximating the absolute value or sign functions in $f$. 
For example, let $\nu_t(y)$  be an approximation of the absolute value function $|y|$, with parameter $t$ controlling the approximation error. For example one could use $\nu_t(y)=y \tanh(t y)$ or $\nu_t(y)=y \text{erf}(t y)$ or a best polynomial approximation of degree $t$. Then taking $f(p,q)=f_t(p,q) = \frac{1}{2q}  \nu_t(p-q)$ gives a smoothed total variation $D_f(p,q) = \text{TV}^*(p,q)$ with
\begin{equation*}
\text{TV}^*\Big(p_a(y) , g(y;\beta)\Big) = \frac{1}{2} \int  \nu_t \Big\{ p_a(y) - g(y;\beta) \Big\} \ dy .
\end{equation*} 
There exist polynomial and rational approximations $\nu_t(y)$ of degree $t$ ensuring that $| \text{TV}(p,q) - \text{TV}^*(p,q)|$ is of order $t^{-1}$ and $\exp(-t)$, respectively \citep{newman1964rational}.  
We also note that the Hellinger divergence is closely related to total variation in the sense that $H^2(p,q)/2  \leq  \text{TV}(p,q) \leq   H(p,q)$ for any densities $p,q$. \\
\end{example}

Figure \ref{fig:projex} shows a few projections of a true density onto
a truncated trigonometric series with six terms, using four different distances ($L_2$, Kullback-Leibler, $\chi^2$, and Hellinger). The projections are all very similar in both cases. However, we note that, as discussed for example in \citet{beran1977minimum}, Hellinger projections should be more stable and robust to outliers or contamination, compared to for example KL. The projections are closer to the true density for the first simpler Gaussian mixture, and are more of a rough approximation for the second more complex mixture.  \\

\begin{figure}[h!]
\begin{center}
{\includegraphics[width=.49\textwidth]{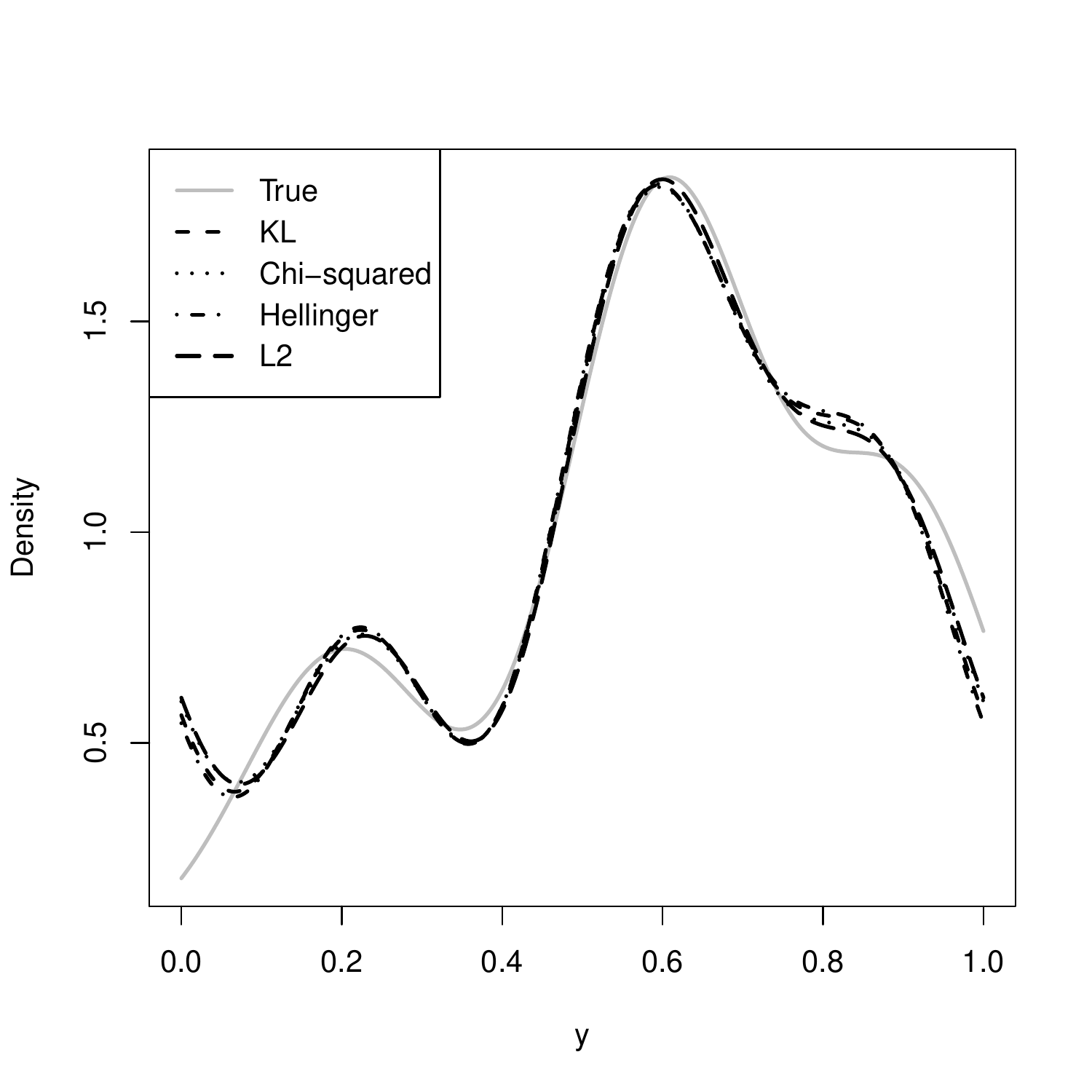}}
{\includegraphics[width=.49\textwidth]{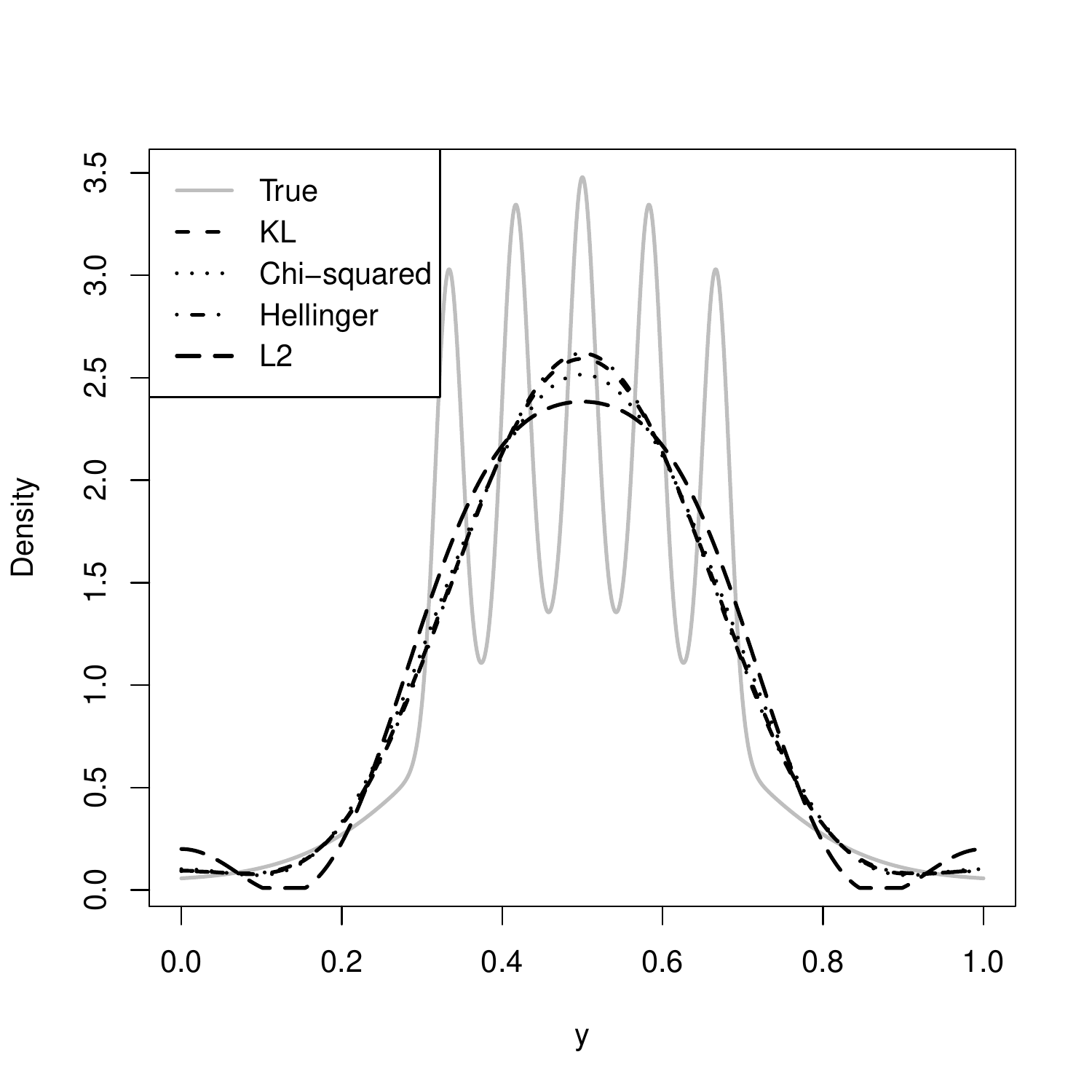}}
\caption{\em Projections of a truncated Gaussian mixture (left) and the
Bart Simpson density (right) onto a trigonometric basis with six
terms, using $L_2$ distance, along with Kullback-Leibler, $\chi^2$, and Hellinger
divergences. } \label{fig:projex}
\end{center}
\end{figure}

\subsubsection{Moment Conditions}
\label{sec:densityfns_momcon}

The next proposition shows how, for smooth distances, the projection parameter $\beta_0$ 
can be defined more explicitly than in equation \eqref{eq:target}, as a solution to a population moment condition, 
involving derivatives of the model $g(y;\beta)$ and the
function $f$. This links projection parameters to integral functionals of the counterfactual density (i.e., moments of transformations of counterfactuals), which is why our efficiency bounds and estimators in the next section resemble those for means of particular non-trivial functions of counterfactuals.  \\

\begin{proposition} \label{prop:mb}
Assume $g$ is differentiable in $\beta$, $f$ is differentiable in its second argument with derivative $f'_2(q_1,q_2)=\frac{\partial}{\partial q_2} f(q_1,q_2)$, and that the minimizer in
\eqref{eq:target} is unique. Then the projection parameter 
$$\beta_0 = \argmin_{\beta \in \R^p} \ D_f\Big( p_a(y) , g(y;\beta) \Big)  $$ 
can be expressed as a solution to the moment condition $m(\beta)=0$, where 
\begin{equation}\label{eq:mb}
m(\beta) \equiv \int \frac{\partial g(y;\beta)}{\partial \beta} 
\left\{ f\Big(p_a(y) , g(y;\beta) \Big) + g(y;\beta) f_2'\Big(p_a(y) , g(y;\beta) \Big) \right\} \ dy . 
\end{equation}
\end{proposition}

\bigskip

The proof of Proposition \ref{prop:mb} follows from the chain rule; all subsequent proofs are given in Appendix \ref{sec:proofs}. Throughout we assume there is a unique solution to $m(\beta)=0$.  Next we show how the moment condition defining $\beta_0$ simplifies for particular distances. \\

\begin{corollary} \label{cor:mb}
The quantity $f\Big(p_a(y) , g(y;\beta) \Big) + g(y;\beta) f_2'\Big(p_a(y) , g(y;\beta) \Big) $ in the integrand of the moment \eqref{eq:mb} equals  
\begin{equation*}
 \begin{cases} 
 2 \Big\{ g(y;\beta) - p_a(y) \Big\} & \text{ if } D_f=L_2^2 \\[10pt]
\displaystyle 1 - \frac{p_a(y)}{g(y;\beta)} & \text{ if } D_f=\text{KL} \\[15pt]
\displaystyle 1 - \left\{ \frac{p_a(y)}{g(y;\beta)} \right\}^2  & \text{ if } D_f=\chi^2 \\[15pt]
\displaystyle 1 - \sqrt{\frac{p_a(y)}{g(y;\beta)} }  & \text{ if } D_f=H^2 \\[15pt]
\displaystyle -\nu_t'\Big\{ p_a(y) - g(y;\beta) \Big\}/2 &  \text{ if } D_f=\text{TV}^* .
\end{cases}
\end{equation*}
\end{corollary}

\bigskip

Corollary \ref{cor:mb} shows how the moment $m(\beta)$
essentially reduces to functionals of the counterfactual density for
particular distances: simple means for $L_2^2$ and KL, a quadratic
functional for $\chi^2$, and a square root functional for $H^2$. For the smoothed TV distance, it depends on the form of the absolute value approximation (e.g., for $\nu_t$ a $t$ degree polynomial approximation, the moment $m(\beta)$ would be an integral of a $t-1$ degree polynomial in the countef).  \\

In the following corollaries we show how the form of the moment condition is particularly straightforward when based on $L_2$ or KL divergence with series models and exponential families, respectively. \\

\begin{corollary} \label{cor:l2mbeta}
If $D_f=L_2^2$ then 
$$ m(\beta) = 2 \int \frac{\partial g(y;\beta)}{\partial \beta} \Big\{ g(y;\beta) - p_a(y) \Big\} \ dy . $$
Therefore if the support of $Y$ is $[0,1]$, and $g(y;\beta) = 1 + \beta^\T b(y)$ is the truncated series in Example \ref{ex:series} then
\begin{equation}
\beta = \E\Big\{ b(Y^a) \Big\} \label{eq:l2beta}
\end{equation}
when $b(\cdot)$ is an orthogonal series with $\int b_j(y) \ dy = 0$ and $\int b_j(y) b_k(y) \ dy = \one(j=k)$.
\end{corollary}

\bigskip

Corollary \ref{cor:l2mbeta} shows that when using orthogonal series with $L_2^2$ projections, there is a closed form for $\beta$, given by a simple mean of a known function of the counterfactual outcome. Estimation and inference for parameters like \eqref{eq:l2beta} is relatively well-understood \citep{robins2009semiparametric, robins2017minimax}, which allows exploiting existing theory and methods in the density estimation context. \\

\begin{corollary} \label{cor:klmbeta}
If $D_f=\text{KL}$ then
\begin{equation*}
m(\beta) = - \E\left\{ \frac{\partial}{\partial \beta} \log g(Y^a;\beta) \right\} ,
\end{equation*}
and so if $g(y;\beta)= \exp\{ \beta^\T b(y) - C(\beta) \}$ is the exponential family in Example \ref{ex:expfam} then
\begin{equation}
m(\beta) =   \frac{\partial}{\partial \beta}C(\beta) - \E\Big\{ b(Y^a) \Big\} = \int b(y) \Big[ \exp\{ \beta^\T b(y) - C(\beta) \} - p_a(y) \Big] \ dy.
\end{equation}
\end{corollary}

\bigskip

Similarly, for KL divergence, the moment $m(\beta)$ is simply the expected score under counterfactual density $g(y;\beta)$. Therefore, just as in the non-counterfactual setting, the parameter values that maximize a posited likelihood are also those that minimize KL divergence \citep{huber1967behavior, white1982maximum}. When one also uses an exponential family, the solution to $m(\beta)=0$ corresponds to an intuitive ``moment matching'', i.e., finding the value of $\beta$ that equates  expectations of $b(\cdot)$ under $g$ to those under the distribution of $Y^a$. \\

\subsection{Distances \& Density Effects}
\label{sec:denseff}

In addition to estimating projections of the counterfactual density
onto a finite-dimensional model, in this section we also consider
estimation of distributional distances themselves.
The main focus is on density-based effects measuring the distance
between counterfactual densities in terms of $L_p^p$ and $f$-divergences. These
effects can detect more nuanced disinctions between the distributions
of $Y^1$ and $Y^0$, beyond simple differences-in-means captured by
standard average treatment effects. \\

More specifically, we consider the distance between $p_1$ and $p_0$
given by
\begin{equation}
\psi_f = D_f \Big(p_1(y) , p_0(y)\Big) = \int f\Big( p_1(y) , p_0(y) \Big) p_0(y) \ dy  \label{eq:denseff}
\end{equation}
for discrepancy functions $f$ as discussed in the previous subsection. 
In this setup we do not require approximating the densities $p_a(y)$
with finite-dimensional models, and instead consider estimating
$\psi_f$ in a fully nonparametric model. \\

\subsection{Model Selection \& Aggregation} 
\label{sec:modsel}

In practice one may not have an approximating model such as \eqref{eq:expfam} available \emph{a priori}. In these cases it would be natural to instead set up a sequence of models, and use the data to choose among them. In standard regression and density estimation problems, simple cross-validation procedures are available for this task; however, because our goal is estimation of a more nuanced counterfactual density, these require some refinement, in the same spirit as \citet{van2003cross}. Thus in this section we describe how the target quantities of Sections \ref{sec:densityfns} and \ref{sec:denseff} can be adapted for the purposes of model selection and aggregation. \\

Specifically, for a set of estimators $\{ \widehat{g}_k(y) : k=1,...,K \}$ of $p_a(y)$ (e.g., estimated from some initial training sample, with each projected onto the space of valid densities), we can define the risk for a given estimator as 
\begin{equation}
R(\widehat{g}_k) = D_f\Big( p_a(y) , \widehat{g}_k(y) \Big)
\end{equation}
The minimum risk oracle estimator $\widehat{g}_{k_0}(y)$ can then be defined via
\begin{equation}
k_0 = \argmin_k R(\widehat{g}_k) = \argmin_k D_f\Big( p_a(y) , \widehat{g}_k(y)\Big) . \label{eq:modelselection}
\end{equation}
A model aggregation oracle can be defined more generally as $\widetilde{g}(y) = \sum_k \beta_{0k} \widehat{g}_k(y)$ where
\begin{equation} \label{eq:aggregation}
\beta_0 = \argmin_{\beta \in B} D_f \left( p_a(y) ,  \sum_{k=1}^K \beta_k \widehat{g}_k(y)  \right) .
\end{equation}
for some appropriate selection set, e.g., the standard simplex $B=\{(\beta_1,...,\beta_K) \in \R^K : \beta_k \geq 0, \sum_k \beta_k=1\}$ for convex aggregation \citep{tsybakov2003optimal, rigollet2007linear}. If one takes $B=\R^K$ for linear aggregation, then $f$-divergences may not be well-defined, so this might naturally only be used in the $D_f=L_2^2$ setting. \\

Note that the proposed target parameters in Section \ref{sec:densityfns} correspond to the aggregation target in \eqref{eq:aggregation} if we replace $\R^d$ with the relevant space $B$. However, since  model selection as defined in Equation \eqref{eq:modelselection} does not satisfy the smoothness assumptions we relied on in Section \ref{sec:densityfns_momcon}, it can be useful in practice to estimate the risk separately for all $K$ candidates; this is more akin to the effect estimation problem in Section \ref{sec:denseff}, except where the density $p_0(y)$ in \eqref{eq:denseff} is replaced with a candidate estimator $\widehat{g}_k(y)$ (e.g., which may be estimated on a separate independent sample/fold and conditioned upon, and so treated as fixed).  \\

\section{Efficiency Theory}
\label{sec:efficiency}

In this section we present a crucial von Mises expansion (i.e., distributional Taylor expansion) for generic density functionals, which yields efficient influence functions for the projection parameters and density effects of interest, and thus nonparametric efficiency bounds \citep{bickel1993efficient, van2003unified}. The latter can be further formalized as local minimax lower bounds \citep{van2002semiparametric}. \\

Throughout we make reference to the linear map $T \mapsto \phi_a(T;\Pb)$ defined as
\begin{equation} \label{eq:eifmap}
\phi_a(T;\Pb) = \frac{\one(A=a)}{\pi_a(X)} \Big\{ T - \E(T \mid X, A=a) \Big\} +  \E(T \mid X, A=a) - \E\{ \E(T \mid X, A=a) \}
\end{equation}
which takes a random variable $T$ (and distribution $\Pb$) and outputs the  efficient influence function for the functional $ \E\{ \E(T \mid X, A=a) \}$. Note we drop the dependence of $\phi_a(T;\Pb)$  on $(X,A)$ for simplicity; at times we also drop the dependence on $\Pb$ if the context is clear. In all our examples, $T=h(Y)$ will be a known or $\Pb$-dependent function of $Y$; the functionals we consider all have influence functions consisting of terms of the above form, but with different and non-standard choices of $T=h(Y)$, depending on the model and distance being used. \\

Recall that in Corollary \ref{cor:mb} we showed the relevant moment $m(\beta)$ reduces to a functional of the counterfactual density for particular distances. Therefore our first result  gives a von Mises-style expansion for generic smooth integral functionals of the counterfactual density. This result paves the way for later expansions and efficiency bounds, and may be of independent interest in other problems involving different counterfactual density functionals. \\

\begin{lemma} \label{lem:densfunc}
Let $\psi = \psi(\Pb) = \int h(p_a(y)) \ dy$ for 
some twice continuously differentiable function $h$. Then $\psi$ satisfies the von Mises expansion
\begin{equation}
\psi(\overline\Pb) - \psi(\Pb) = \int \phi_a\left(h'\Big(p_a(Y)\Big);\overline\Pb \right) \ d(\overline\Pb-\Pb) + R_2(\overline\Pb,\Pb)
\end{equation}
where
\begin{align*}
R_2(\overline\Pb,\Pb) &= \int \int h'(\overline{p}_a(y)) \left\{ \frac{\pi_a(x) }{\overline\pi_a(x)} - 1 \right\}  \Big\{ \eta_a(y \mid x) - \overline\eta_a(y \mid x) \Big\} \ dy \ d\Pb(x) \\
& \hspace{.4in} + \frac{1}{2} \int h''(p^*_a(y)) \Big\{ \overline{p}_a(y) - p_a(y) \Big\}^2 \ dy , 
\end{align*}
where $p^*_a(y)$ lies between $p_a(y)$ and $\overline{p}_a(y)$. \\
\end{lemma}

Lemma \ref{lem:densfunc} has several important consequences. First, it
indicates how one can correct the first-order bias of a plug-in
estimator $\psi(\widehat\Pb)$ of counterfactual density functionals:
by estimating the first term in the expansion and subtracting it
off. This is how standard semiparametric estimators (particularly of
the one-step variety) based on influence functions are constructed
\citep{bickel1993efficient, van2003unified, chernozhukov2018double}, and our proposed
estimators in the next section do precisely this.  Second, 
since the remainder term is quadratic
in the nuisance functions,
it implies that $\psi(\Pb)$ is pathwise differentiable 
with efficient influence function $\phi_a(h'(p_a(Y))$;
for this fact we refer to Lemma \ref{lem:exptoeif} in the Appendix. \\

\subsection{Density Functions}
\label{sec:densityfnseif}

In this subsection we use Lemma \ref{lem:densfunc} to detail the efficient influence function for the moment $m(\beta)$ at a fixed $\beta$, as well as the projection parameter $\beta_0$ and projected density $g(y;\beta_0)$. These efficient influence functions yield local minimax lower bounds, as well as estimators that can attain the nonparametric efficiency bounds under generic high-level rate conditions on nuisance estimators, which will be proved in Section \ref{sec:estimation}. \\

\begin{theorem} \label{lem:eifmb}
Assume $f$ is twice differentiable and denote partial derivatives as $f'_j(q_1,q_2)=\frac{\partial}{\partial q_j} f(q_1,q_2)$ and similarly $f''_{jk}(q_1,q_2)=\frac{\partial^2}{\partial q_j \partial q_k} f(q_1,q_2)$. 
Then, under an unrestricted nonparametric model, the efficient influence function for $m(\beta)$ is given by
\begin{align*}
\phi_a\Big(\gamma_f(Y;\beta)\Big)
\end{align*}
where
\begin{align*}
\gamma_f(y;\beta) & \equiv  \gamma_f(y;\beta,p_a) =
\frac{\partial g(y;\beta)}{\partial \beta} \left\{ f'_1\Big( p_a(y), g(y;\beta) \Big) + g(y;\beta) f''_{21}\Big( p_a(y) , g(y;\beta) \Big)   \right\} .
\end{align*}
The efficient influence functions for $\beta_0$ and $g(y;\beta_0)$ are similarly given by
\begin{equation}\label{eq::eff}
- \frac{\partial m(\beta)}{\partial \beta}^{-1} \phi_a\Big(\gamma_f(Y;\beta)\Big) \Bigm|_{\beta=\beta_0}
\ \text{ and } \ 
-\frac{\partial g(y;\beta)}{\partial \beta^\T} \frac{\partial m(\beta)}{\partial \beta}^{-1}\ \phi_a\Big(\gamma_f(Y;\beta)\Big) \Bigm|_{\beta=\beta_0}
\end{equation}
respectively. \\
\end{theorem}

The efficient influence functions given in Theorem \ref{lem:eifmb} are analogous to those of usual ATE-type parameters, but with the crucial difference that they correspond  to means of $\gamma_f(Y^a;\beta)$, not $Y^a$ itself. This is what we should expect based on the result in Lemma \ref{lem:densfunc}, since the $\gamma_f$ transformation is the derivative of the integrand in the moment condition \eqref{eq:mb} given in Proposition \ref{prop:mb}. Note also that the form of $\gamma_f$ indicates that the efficiency bound for $\beta_0$ (i.e., the variance of the efficient influence function) will be adversely affected when the model $g$ is sensitive to small changes in $\beta$, or when the distance is sensitive to small changes in its arguments, since then the derivatives in $\gamma_f$ will be large.  \\

In the next corollary, we give the particular form of the efficient influence functions when $D_f$ is the $L_2^2$ and KL divergence, and the approximating models are a linear series and exponential family. \\

 \begin{corollary} \label{cor:eif_l2kl}
 For $L_2^2$ and KL divergence the quantity $\gamma_f$ from Theorem \ref{lem:eifmb} reduces to
 $$ \gamma_f(y;\beta) = \begin{cases}
 - 2 \frac{\partial g(y;\beta)}{\partial \beta}  & \ \text{ if } D_f = L_2^2 \\
  -  \frac{\partial \log g(y;\beta)}{\partial \beta}  & \ \text{ if } D_f = \text{KL} .
 \end{cases} $$
Further, if either
 \begin{enumerate}
 \item $D_f=L_2^2$ and $g(y;\beta) = q(y) + \beta^\T b(y)$ is the truncated series in Example \ref{ex:series}, or
 \item $D_f=\text{KL}$ and $g(y;\beta)= \exp\{ \beta^\T b(y) - C(\beta) \}$ is the exponential family in Example \ref{ex:expfam}
 \end{enumerate}
 then the efficient influence function for $m(\beta)$ is proportional to
 \begin{align*}
\phi_a \Big(b(Y) \Big) .
\end{align*}
The proportionality constant is $-2$ for $D_f=L_2^2$, and $-1$ for $D_f=\text{KL}$. 
 \end{corollary}

Corollary \ref{cor:eif_l2kl} shows that the efficient influence functions are proportional for linear projections using $L_2^2$ distance, and for projections onto an exponential family using the KL divergence. Further, this efficient influence function simply corresponds to that of the counterfactual mean vector $\E\{b(Y^a)\}$, for $b$ a known basis function vector. Thus the influence function conveniently reduces to that of the mean of a transformed version of the counterfactual outcome, with no dependence on $\beta$. As mentioned after Corollary \ref{cor:l2mbeta}, this allows for adapting existing theory and methods for average treatment effects to the density estimation context. \\

The following theorem summarizes the local minimax lower bound implied by the form of the efficient influence function in Theorem \ref{lem:eifmb}, as in Corollary 2.6 of  \citet{van2002semiparametric}.

\begin{corollary} \label{cor:minimax}
Let $\sigma^2=\sigma_\Pb^2$ denote the variance of the efficient influence
function from (\ref{eq::eff}).
The local minimax risk for $\beta_0$ is lower bounded
as
$$ \inf_{\delta>0} \ \liminf_{n \rightarrow \infty} \ 
\sup_{\text{TV}(\overline\Pb,\Pb)<\delta} \ 
\E_{\overline\Pb} \left[ \ell \left\{ \sqrt{n} \Big( \widehat\beta - \beta_0(\overline\Pb) \Big) \right\} \right] \geq 
\E \Big\{ \ell(\sigma Z) \Big\}  $$
for any estimator $\widehat\beta$, where $\ell: \R^p \mapsto [0,\infty)$ is any subconvex loss function. 
\end{corollary}

Corollary \ref{cor:minimax} follows from Corollary 2.6 of \citet{van2002semiparametric}. It shows that the worst-case mean squared error of any estimator, locally near the true $\Pb$, cannot be smaller than the efficiency bound, asymptotically and after scaling by $\sqrt{n}$.   This   gives an important benchmark for efficient estimation of projection parameters of the counterfactual density: no estimator can have mean squared error uniformly better than the variance of the efficient influence function (divided by $n$), without adding extra assumptions to the nonparametric model we consider. 

\bigskip

\subsection{Density Effects}
\label{sec:denseffeif}

Now we give the efficient influence function for the density effect parameters in \eqref{eq:denseff}. Unlike the projected densities in the previous subsection, the density effect parameters depend on both counterfactual densities of interest for comparison. \\

\begin{theorem} \label{thm:eifdenseff}
In an unrestricted nonparametric model, the efficient influence function for the density effect $\psi_f=\int f\left( {p_1(y)} , {p_0(y)} \right) p_0(y) \ dy $ is given by
\begin{align*}
\phi_1 \Big( \lambda_{1}(Y) \Big)  + \phi_0\Big( \lambda_0(Y) \Big) 
\end{align*}
where
\begin{align*}
\lambda_1(y) &= p_0(y) f'_1\Big(p_1(y),p_0(y) \Big) \\
\lambda_0(y) &= f\Big(p_1(y) , p_0(y) \Big) + p_0(y) f_2'\Big(p_1(y) , p_0(y) \Big) .
\end{align*}
\end{theorem}

\bigskip

As with the result for $\beta_0$ in Theorem \ref{lem:eifmb}, the efficient influence function for $\psi_f$ in Lemma \ref{thm:eifdenseff} consists of inverse probability weighted residuals, plus a ``plug-in''-type term, similar to ATE parameters. However, again this corresponds to the influence function for a transformed version of the outcome, depending on the counterfactual densities and choice of distance $f$. The efficient influence function simplifies somewhat for $L_2^2$ and KL divergence, as indicated in the following corollary. Expressions for other $f$-divergences are in Section \ref{sec:fexamples} in the Appendix.  \\

 \begin{corollary} \label{thm:eifdenseff_kl}
  If $D_f=L_2^2$, then the efficient influence function for $\psi_f$ is 
$$ 2 (\phi_1 - \phi_0) \Big( p_1(Y) - p_0(Y) \Big)  . $$
 If $D_f=\text{KL}$, then the efficient influence function for $\psi_f$ is 
 $$ \phi_1\left( \log \left( \frac{p_1(Y)}{p_0(Y)} \right) \right) - \phi_0 \left( \frac{p_1(Y)}{p_0(Y)}  \right) . $$
 \end{corollary}

\bigskip

The fact that $\lambda_1=-\lambda_0$ for $L_2^2$ projections simplifies the form of our proposed estimators, as we will detail further in the next section. We also note that the influence function  reduces to zero when $p_1=p_0$, which presents some complications for inference; this will be discussed in the next section as well. \\ 

As mentioned in Section \ref{sec:modsel}, for the purposes of model selection and aggregation it is also useful to consider the distance between $p_a$ and a fixed candidate $g$; we give the corresponding efficient influence function here. \\

\begin{proposition} \label{prop:eifmodsel}
In an unrestricted nonparametric model, the efficient influence function for 
$\Delta_f(g)=\int f\left( {p_a(y)} , {g(y)} \right) g(y) \ dy$ for $g$  fixed and known is given by
$$ \phi_a\left( g(Y) f_1'\Big( p_a(Y), g(Y) \Big) \right).  $$
If $D_f=L_2^2$ then this influence function reduces to 
$$ 2 \phi_a\Big( p_a(Y) - g(Y) \Big) . $$
\end{proposition}

\medskip

\section{Estimation and Inference}
\label{sec:estimation}

In this section we present doubly robust-style estimators of the proposed density functions and density effects, based on the functional expansions from Lemma \ref{lem:densfunc} and the efficient influence function results in Theorems \ref{lem:eifmb}--\ref{thm:eifdenseff}. We study their rates of convergence, and show they can be $n^{-1/2}$ consistent and 
asymptotically efficient under weak nonparametric conditions. \\

\subsection{Density Functions}
\label{sec:densfnest}

Here let $\hat \pi_a(x)$ and
$\hat\eta_a(y \mid x)$ denote initial estimators of the propensity score and 
conditional density functions $\pi_a(x) = \Pb(A=a \mid X=x)$ and 
$\eta_a(y \mid x) = \frac{\partial}{\partial y} \Pb(Y \leq y \mid X=x, A=a)$, for example 
based on generic regression estimators and their numerical derivatives (or for the latter one can use a regression of a kernel transformed version of the outcome). 
Also let $\hat p_a(y) = \Pn\{ \widehat\eta_a(y \mid X)\}$ denote the plug-in estimator 
of the counterfactual density under $A=a$, 
where
$\Pn\{ h(Z)\} = n^{-1}\sum_i h(Z_i)$,
and let 
\begin{equation}\label{eq:mbplugin}
\widehat{m}(\beta) \equiv \int \frac{\partial g(y;\beta)}{\partial \beta} 
\left\{ f\Big(\widehat{p}_a(y) , g(y;\beta) \Big) + g(y;\beta) f_2'\Big(\widehat{p}_a(y) , g(y;\beta) \Big) \right\} dy . 
\end{equation}
denote the plug-in estimator of the moment condition $m(\beta)$, and similarly for $\psi_f$. \\

\begin{remark} \label{rem:initdens}
Although we suggest basing \eqref{eq:mbplugin} on the plug-in estimator $\widehat{p}_a(y) = \Pn\{ \widehat\eta_a(y \mid X)\}$ of the counterfactual density, one could just as well use other estimators (e.g., inverse-probability-weighted, or doubly robust, as in \citet{kim202xcausal}). 
 Nonetheless, all results in this paper will only depend on high-level second-order rate conditions for estimating $p_a(y)$, which would be satisfied for the simple plug-in estimator as long as similar conditions hold for the underlying density estimator $\widehat\eta_a(y \mid x)$.  We prove this in Appendix \ref{app:initdens}, showing that the mean squared error of $\widehat{p}_a(y)$ is upper bounded by an integrated version of that of $\widehat\eta_a(y \mid x)$.  \\
\end{remark}

To ease notation we let $\widehat\phi_a(T) = \phi_a(T;\widehat\Pb)$ denote the estimated
version of the efficient influence function   given in \eqref{eq:eifmap}.
Then our proposed  projection estimators  are
given by approximate solutions in $\beta$ (up to $o_\Pb(1/\sqrt{n})$ error) to 
\begin{equation}
\widehat{m}(\beta) + \Pn \left\{ \widehat\phi_a\Big(\widehat\gamma_f(Y;\beta)\Big) \right\} = o_\Pb(1/\sqrt{n}) \label{eq:onestep}
\end{equation}
In other words the estimators are  one-step bias-corrected estimators (of the moment condition and the parameter itself, respectively), which take the plug-in estimator 
and add  an estimate of the bias by averaging an estimate of the influence function. \\

\begin{remark} \label{rem:splitting}
For simplicity, in the following results we assume the various nuisance estimates in $\widehat\Pb$ are constructed
from a single separate independent sample, of the same size $n$ as the estimation sample on which $\Pn$ operates. 
Alternatively, if the same observations
are used both for estimating nuisance functions and averaging estimates of the influence function, one generally needs to rely on empirical
process conditions to avoid overfitting.
In practice, with iid data, one can always obtain separate independent samples by randomly
splitting the data in half (or in folds); further, to regain full sample size efficiency one can
always swap the samples, repeat the procedure, and average the results, popularly called cross-fitting and used for example by 
\citet{bickel1988estimating, chernozhukov2018double, robins2008higher, zheng2010asymptotic}. 
 In this paper, to simplify notation we always
analyze a single split procedure, with the understanding that extending to an analysis of an
average across independent splits is straightforward. \\
\end{remark}

Our first propositions give the form of the plug-in and bias-corrected projection estimators when using a linear series with $L_2^2$ distance, and an exponential family model with KL divergence, which take a particularly simple form. 

\bigskip

\begin{proposition} \label{prop:estexpl2}
If $D_f=L_2^2$, the support of $Y$ is $[0,1]$, and $g(y;\beta) = 1 +
\beta^\T b(y)$ is the truncated series in Example \ref{ex:series},
with $b(\cdot)$ an orthogonal series with $\int b_j(y) \ dy = 0$ and
$\int b_j(y) b_k(y) \ dy = \one(j=k)$, then the plug-in estimator of
$\beta$ is
\begin{equation*}
\widehat\beta = \Pn \{ \widehat\mu_a(X;b) \} , 
\end{equation*}
where $\widehat\mu_a(x;b)$ is an estimate of $\mu_a(x;b)=\E\{b(Y) \mid
X=x, A=a)$. In contrast, the proposed one-step estimator in
\eqref{eq:onestep} is given by 
\begin{equation}
\widehat\beta = \Pn \left[ \frac{\one(A=a)}
{\widehat\pi_a(X)} \Big\{ b(Y) - \widehat\mu_a(X;b) \Big\} + 
\widehat\mu_a(X;b) \right] .
\end{equation}
\end{proposition}
 
\bigskip
 
\begin{proposition} \label{prop:estexpkl}
If $D_f=\text{KL}$ and $g(y;\beta)= \exp\{ \beta^\T b(y) - C(\beta)\}$ is the exponential family in Example \ref{ex:expfam}, then the
plug-in estimator solving $\widehat{m}(\widehat\beta)=0$ is the
solution in $\beta$ to
\begin{equation*}
\int \Big[ b(y) - \Pn \{ \widehat\mu_a(X;b) \} \Big] \exp\Big\{ \beta^\T b(y)   \Big\}  \ dy = 0
\end{equation*}
where $\widehat\mu_a(x;b)$ is an estimate of $\mu_a(x;b)=\E\{b(Y) \mid
X=x, A=a)$. In contrast, the proposed one-step estimator in
\eqref{eq:onestep} is given by the solution in $\beta$ to
\begin{equation} \label{est:estexpkl}
\int \left( b(y) - \Pn \left[ \frac{\one(A=a)}
{\widehat\pi_a(X)} \Big\{ b(Y) - \widehat\mu_a(X;b) \Big\} + 
\widehat\mu_a(X;b) \right] \right) \exp\Big\{ \beta^\T b(y) \Big\}  \ dy = 0 . 
\end{equation}
 \end{proposition}
 
 \bigskip
 
 Propositions  \ref{prop:estexpl2}-\ref{prop:estexpkl} shows that the plug-in and bias-corrected estimators for $L_2^2$ and KL projections  solve simple estimating equations, which only require one to first estimate the components $\E\{ \mu_a(X;b) \}$; importantly, straightforward doubly robust estimators as in \eqref{est:estexpkl} are available, and do not depend on the estimating equation parameter $\beta$. This is not necessarily true for other model/distance combinations; in general $\widehat\gamma_f$ would have to be estimated at each $\beta$ in order to solve \eqref{eq:onestep}, which could be quite computationally intensive.  \\
 
Next we give the main result of this section, which shows the rate of convergence for the proposed estimator. Importantly the rate involves products of nuisance estimation errors, allowing for $n^{-1/2}$ consistency and asymptotic normality in nonparametric models, and even when the nuisance estimators are generic and flexibly fit. \\

\begin{theorem}\label{thm:densfnrate}
Let  $\eta = (\pi_a,\eta_a)$, and $\varphi(Z;\beta,\eta) = m(\beta;\eta) + \phi_a(\gamma_f(Y;\beta),\eta)$. Assume:
\begin{enumerate}
\item The functions $\gamma_f$ and $1/\widehat\pi_a$ are bounded above by some constant, and $\gamma_f$ is differentiable in $p_a(y)$, with derivative bounded uniformly above by $\delta$. 
\item The function class $\{ \varphi(z;\beta,\eta) : \beta \in \R^p\}$ is Donsker in $\beta$ for any fixed $\eta$.
\item The estimators are consistent in the sense that $\widehat\beta-\beta_0 = o_\Pb(1)$ and $\| \widehat\eta - \eta_0 \| = o_\Pb(1)$.
\item The map $\beta \mapsto \Pb\{ \varphi(Z;\beta,\eta)\}$ is differentiable at $\beta_0$ uniformly in $\eta$, with nonsingular derivative matrix $\frac{\partial}{\partial\beta} \Pb\{\varphi(Z;\beta,\eta)\} |_{\beta=\beta_0} = V(\beta_0,\eta)$, where $V(\beta_0,\widehat\eta) \inprob V(\beta_0,\eta_0)$.
\end{enumerate}
Then 
\begin{align*}
\widehat\beta - \beta_0 &= -V(\beta_0,\eta_0)^{-1} (\Pn-\Pb)  \left\{ \phi_a\Big(\gamma_f(Y;\beta_0)\Big) \right\} \\
& \hspace{.5in} + O_\Pb\left(\| \widehat\pi_a - \pi_a \| \| \widehat\eta_a - \eta_a \| + \delta \| \widehat{p}_a - p_a \|^2  + o_\Pb\left( \frac{1}{\sqrt{n}} \right) \right).
\end{align*}
\end{theorem}

\bigskip

\begin{remark}
In a slight abuse of notation, Theorem \ref{thm:densfnrate} holds when we define $\| \widehat\eta_a - \eta_a \|^2 = \| \zeta_a \|^2 \equiv  \int \zeta_a(x)^2 \ d\Pb(x)$ for integrated error $\zeta_a(x) = \int | \widehat\eta_a(y \mid x) - \eta_a(y \mid x)| \ dy$. This implies it also holds if we define $\| \widehat\eta_a - \eta_a \|^2 = \int \{\widehat\eta_a(y \mid x)  - \eta_a(y \mid x) \}^2 \ dy \ d\Pb(x)$, or 
$\| \widehat\eta_a - \eta_a \|^2 = \int \{\widehat\eta_a(y \mid x)  - \eta_a(y \mid x) \}^2 \ d\Pb(y,x)$ if $\eta_a(y \mid x)$ is bounded from below. \\
\end{remark}

Importantly, Theorem \ref{thm:densfnrate} shows that $\widehat\beta$
 attains substantially faster rates than its nuisance estimators $\widehat\eta$,
and can be asymptotically efficient under weak nonparametric conditions, 
for example attaining the minimax lower bound in Corollary \ref{cor:minimax}. 
First we give some description of the assumed conditions. 
The first condition ensures the influence function is not too complex 
as a function of $\beta$ (though allowing arbitrary complexity in $\eta$). 
The second condition merely requires consistency of $(\widehat\beta,\widehat\eta)$ at any rate. 
The third condition requires some smoothness in $\beta$, so as to allow a delta method argument. 
These conditions ensure $\widehat\beta$ has a rate of 
convergence that is second-order in the nuisance estimation error, thus attaining faster rates
than the nuisance estimators. Thus, for example, under standard $n^{-1/4}$-type rate conditions on $\widehat\eta$, 
the estimator $\widehat\beta$ is $n^{-1/2}$-consistent, asymptotically normal, and efficient. 
Importantly, these rates can be attained under smoothness, sparsity, or other structural conditions (e.g.,
additive modeling or bounded variation assumptions, etc.). For instance, if it is assumed that
all $d$-dimensional nuisance functions lie in a Holder class with smoothness index $s$ (i.e., partial
derivatives up to order $s$ exist and are Lipschitz) then the assumption of Theorem \ref{thm:densfnrate} would be
satisfied when $s > d/2$, i.e., the smoothness index is at least half the dimension. Alternatively,
if the functions are $s$-sparse then one would need $s = o(\sqrt{n})$ 
 up to log factors, as in \citet{farrell2015robust}. In these cases, asymptotically valid 95\% confidence intervals can be 
 constructed via the simple Wald form, $\widehat\beta \pm 1.96 \sqrt{\text{diag}[ \widehat\cov\{ \widehat\phi_a(\widehat\gamma_f(Y;\widehat\beta))\}/n]}$.  \\
 
 \begin{remark}
 In some prominent cases (for example, $L_2^2$ and KL projections, as shown in Corollary \ref{cor:eif_l2kl}), the function $\gamma_f$ does not depend on the counterfactual density $p_a(y)$ at all, so its derivative is exactly zero and $\delta=0$. 
In this case the second term in the second-order remainder in Theorem \ref{thm:densfnrate} drops out, making the proposed approach doubly robust in the usual sense, requiring no rate conditions on the initial pilot estimate of the counterfactual density. \\
 \end{remark}

\subsection{Density Effects}
\label{sec:denseffest}

Here we present doubly robust-style estimators of the density effects described in Section \ref{sec:denseff}, and study their rate of convergence.  
As before we first construct initial estimators $\hat \pi_a(x)$, 
$\hat\eta_a(y \mid x)$, and $\widehat{p}_a(y)=\Pn\{ \widehat\eta_a(y \mid X)\}$ of the propensity score and  
conditional and counterfactual densities. Estimated versions of $\phi_a(T)$ and $\lambda_a$
defined in Theorem \ref{thm:eifdenseff} follow accordingly. \\

Then the density effect estimators we  propose are defined as
\begin{equation}
\widehat\psi_f = \int f\Big( \widehat{p}_1(y) , \widehat{p}_0(y) \Big) \widehat{p}_0(y) \ dy + \Pn \left\{ \widehat\phi_1 \Big( \widehat\lambda_{1}(Y) \Big)  + \widehat\phi_0\Big( \widehat\lambda_0(Y) \Big)  \right\} ,  \label{eq:onestep_psi}
\end{equation}
which can again be viewed as one-step bias-corrected estimators, with plug-in bias estimated via an average of the estimated influence function. In practice, rather than estimating the conditional density $\eta_a$ and integrating over its $y$ argument, one could instead regress for example $\widehat\lambda_a$ on $X$ for the integral terms in the estimated influence function. \\

\begin{proposition} \label{prop:l2denseff}
If $D_f=L_2^2$ then the proposed density effect estimator can be written as
\begin{align*}
\widehat\psi_f &=  2 \ \Pn \bigg(  \frac{2A-1}{\widehat\pi_A(X)} \left[ \Big\{ \widehat{p}_1(Y) - \widehat{p}_0(Y) \Big\} -  \int \Big\{ \widehat{p}_1(y) - \widehat{p}_0(y) \Big\} \widehat\eta_A(y \mid X) \ dy  \right] \\
& \hspace{.4in} + \int \Big\{ \widehat{p}_1(y) - \widehat{p}_0(y) \Big\} \Big\{ \widehat\eta_1(y \mid X) - \widehat\eta_0(y \mid X)  \Big\} \ dy \bigg) -   \int \Big\{ \widehat{p}_1(y) - \widehat{p}_0(y) \Big\}^2 \ dy .
\end{align*}
\end{proposition}

\medskip

The estimator in Proposition \ref{prop:l2denseff} can be viewed as taking twice the doubly robust estimator of the mean of $(\widehat{p}_1(Y^1) - \widehat{p}_0(Y^1)) - (\widehat{p}_1(Y^0) - \widehat{p}_0(Y^0))$, which is $\int (\widehat{p}_1-\widehat{p}_0)(p_1-p_0)$, and subtracting a plug-in estimate of the $L_2^2$ distance. This is analogous to the standard one-step estimator of the expected (observational) density $\int p(x)^2 \ dx$ \citep{bickel1988estimating}, which takes twice an estimate of the mean of $\widehat{p}(X)$, i.e., $\int \widehat{p} p$, and subtracts the plug-in estimate $\int \widehat{p}^2$. For the expected density, the bias is just the integrated squared difference between $\widehat{p}$ and $p$;  in contrast, in our setting,  we show next that there  is an additional doubly robust error term, due to the confounding adjustment required for estimating counterfactual densities.  \\

\begin{theorem}\label{thm:denseffrate}
Assume $\lambda_a$ and $1/\widehat\pi_a$ are bounded above by some constant for $a=0,1$, and $\lambda_a$ is differentiable in $p_a(y)$, with derivative bounded uniformly above by $\delta_a$. 
Then 
\begin{align*}
\widehat\psi_f - \psi_f &=  (\Pn-\Pb) \left\{ \phi_1 \Big( \lambda_{1}(Y) \Big)  + \phi_0\Big( \lambda_0(Y) \Big)  \right\} \\ 
& \hspace{.5in} 
+ O_\Pb\left( \sum_{a=0}^1 \Big( \| \widehat\pi_a - \pi_a \| \| \widehat\eta_a - \eta_a \| + \delta_a \| \widehat{p}_a - p_a \|^2  \Big) + o_\Pb\left( \frac{1}{\sqrt{n}} \right) \right).
\end{align*}
\end{theorem}

\bigskip

Theorem \ref{thm:denseffrate} (whose proof mimics that of Theorem \ref{thm:densfnrate}) shows that $\widehat\psi_f$ can
 attain faster rates than its nuisance estimators,
and can be asymptotically efficient under weak nonparametric conditions. 
The conditions and the form of the convergence rate  are similar to those of Theorem \ref{thm:densfnrate}, so we refer to our discussion there for more details. However we do comment on a few differences. First, for the density functions targeted in Theorem \ref{thm:densfnrate}, the moment condition $m(\beta)$, and resulting influence functions and estimators, can have a complicated dependence on $\beta$; in contrast, this is not an issue for the density effect $\psi_f$ since the influence function is linear in the parameter. Thus extra smoothness conditions on the influence function  used in Theorem \ref{thm:densfnrate} are not required in Theorem \ref{thm:denseffrate}.  Second, although in Theorem \ref{thm:densfnrate} the derivative bound $\delta$ can be exactly zero in some prominent cases, in general in Theorem \ref{thm:denseffrate} this will not be the case (e.g., for $L_2^2$ distance the derivative of $\lambda_a$ has absolute value equal to one). Therefore, for efficient estimation of density effects, we in general need an initial density estimator converging at $n^{-1/4}$ rate. However recall that, as described in Remark \ref{rem:initdens}, there exist nonparametric counterfactual density estimators with error upper bounded by $\| \widehat\pi_a - \pi_a \|$ or $\| \widehat\eta_a - \eta_a \|$, so that   $\| \widehat{p}_a - p_a \|^2$ would be of smaller or similar order compared to the product error preceding it. \\

There is a third distinction in density effect estimation. 
Under usual $n^{-1/4}$ rate conditions on the nuisance estimators,  Theorem \ref{thm:denseffrate} suggests  95\% confidence intervals of the form
\begin{equation}
\widehat\psi_f \pm 1.96 \sqrt{ \widehat\cov\left\{ \widehat\phi_1 \Big( \widehat\lambda_{1}(Y) \Big)  + \widehat\phi_0\Big( \widehat\lambda_0(Y) \Big)  \right\}/n } \label{eq:denseffci}
\end{equation}
These intervals are asymptotically valid as usual when
$p_1 \neq p_0$, but not when $p_1 = p_0$, since then the influence function
of $\psi_f$ reduces to zero, as mentioned in
Section \ref{sec:denseffeif}. This invalidates inference because the first sample
average in Theorem \ref{thm:denseffrate} is no longer
dominant, as with degenerate U-statistics
or other estimators whose higher-order terms dominate their von Mises
expansions (cf.\ Sections 12.3 and 20.1.1 of
\citet{van2000asymptotic}). However, the presence of
nuisance functions complicates things substantially, as noted in other similarly complex functional estimation
problems \citep{luedtke2019omnibus, williamson2020unified}, but we are
not aware of a general solution. Thus we only recommend using the interval
\eqref{eq:denseffci} in non-null settings when $p_1 \neq p_0$. 
A simple albeit ad-hoc fix is to use the interval
$\widehat\psi \pm z_{\alpha/2} (s\vee 1/\sqrt{n})$
where $s= \sqrt{\widehat\cov\{ \widehat\phi_1 ( \widehat\lambda_{1}(Y) )  + \widehat\phi_0( \widehat\lambda_0(Y) )  \}/n }$.
This is valid but conservative near the null.

 \subsection{Model Selection \& Aggregation}
\label{sec:modselest}

Here we briefly describe how the methods of the previous subsections can be used for the purposes of model selection and aggregation, in the same spirit as \citet{van2003cross, tsybakov2003optimal}, and others. We leave technical details to future work. \\

First we consider the linear aggregation goal as defined in \eqref{eq:aggregation}, where $B=\R^K$. In this setup the methods from Section \ref{sec:densfnest} can be straightforwardly adapted, by adding an extra step of sample splitting. We focus on $L_2^2$ projections since $f$-divergences may not be well-defined for general linear combinations of candidate estimators. Our proposed approach is as follows:
\begin{enumerate}
\setlength\itemsep{0em}
\item[\emph{Step 1.}] Randomly split the sample into a training set $D_n^0$ and test set $D_n^1$. 
\item[\emph{Step 2.}] On the training set $D_n^0$, estimate $K$ different models (e.g., $K$ different numbers of basis functions, or $K$ different combinations of linear, exponential family, Gaussian mixture models, etc.), using the estimator in \eqref{eq:onestep} to compute $\widehat{g}_k(y) = g(y;\widehat\beta_k)$, $k=1,...,K$. 
\item[\emph{Step 3.}] On the test set $D_n^1$, estimate the ($L_2^2$) projection onto an orthonormal basis of the linear span of $(\widehat{g}_1,...,\widehat{g}_K)$, again using the estimator in \eqref{eq:onestep}, e.g., with the series model in Example \ref{ex:series} with $q(y)=0$, to compute an aggregated estimator $\widehat{g}(y) = \sum_k \widehat\theta_k \widehat{g}_k(y)$. 
\item[\emph{Step 4.}] Reverse the roles of $D_n^0$ and $D_n^1$ and average the two resulting aggregates. \\
\end{enumerate}

Note that inside Steps 2-3,  another layer of sample splitting is required to avoid empirical process conditions in estimating the nuisance functions, as discussed in Remark \ref{rem:splitting}. We also note that the cross-fitting in Step 4 could be considered optional if the corresponding efficiency loss was considered negligible, or alternatively one could instead implement Steps 1--4 with $M$ different folds, at each step using $M-1$ for training and the other fold for the test set. We conjecture that the above approach can attain the optimal $K/n$ rates for linear density aggregation in the observational case \citep{rigollet2007linear}, under standard $n^{-1/4}$-type conditions on the nuisance estimators (or weaker, depending on how $K$ scales with $n$). \\

For model selection and convex aggregation, we propose a similiar procedure, except where in Step 3 variants of the density effect estimators from Section \ref{sec:denseffest} are used to estimate the distance between $p_a$ and each of the $k$ candidates estimated from the training split (after projecting each onto the space of valid densities). One can then pick the minimum distance candidate or an appropriately weighted combination, e.g., by finding the convex weights that minimize the estimated distance in the test split. For example, our proposed estimator of the $L_2^2$ error of a candidate $g_k$ based on Proposition \ref{prop:eifmodsel} is given by
$$
\widehat\Delta_f(g_k) = \int \Big( \widehat{p}_a(y) - g_k(y) \Big)^2 \ dy + 2 \Pn \left\{ \widehat\phi_a \Big( \widehat{p}_a(Y) - g_k(Y) \Big)   \right\} . $$
For the purposes of model selection, one can instead use the simpler pseudo-$L_2^2$ risk
\begin{align} \label{eq:pseudol2}
\widehat\Delta^*_f(g_k) &= -2 \ \Pn \bigg[  \frac{\one(A=a)}{\widehat\pi_a(X)} \left\{  g_k(Y)  -  \int  g_k(y)  \widehat\eta_a(y \mid X) \ dy  \right\} \\ 
& \hspace{.75in} + \int g_k(y)   \widehat\eta_a(y \mid X)  \ dy \bigg] + \int   g_k(y)^2 \ dy ,  \nonumber
\end{align}
based on the fact that the $L_2^2$ distance $\int (p_a - g_k)^2$ equals $\int g_k^2 - 2 \int g_k p_a$ plus a term $\int p_a^2$ that does not depend on $g_k$.  
This is the estimator we use in the data analysis in the next section. \\

\section{Illustration}
\label{sec:illustration}

Here we apply our proposed methods to analyze the effect of combined antiretroviral therapy for treating HIV. All code is given in Appendix \ref{sec:code}, and the methods are implemented in the  \emph{npcausal} R package on GitHub (https://github.com/ehkennedy/npcausal). \\

The data we use come from the ACTG 175 randomized trial 
\citep{hammer1996trial}, and are available in the
\verb|speff2trial| R package. The treatment is whether patients
received combination therapy ($A=1$) versus zidovudine alone ($A=0$),
and the outcome $Y$ is CD4 count at 96 weeks post-baseline. Baseline
covariates $X$ include age, weight, Karnofsky score, indicators for
race, gender, hemophilia, homosexual activity, drug use, whether
symptomatic, and previous zidovudine and antiretroviral use. There are
a total of $n=2319$ patients in the trial, 797 of which do not have
outcome data (we use $R=1$ to denote an observed outcome). \\

Since we are interested in
the density of outcomes had all versus none been treated \textit{in
  the entire population} (i.e., had all outcomes been measured), we
can view the product indicator $\one(A=a,R=1)$ as a joint ``treatment''
variable \citep{van2003unified}. In other words our goal is to estimate counterfactual densities under $A=1$ \emph{and} $R=1$, versus  $A=0$ and $R=1$. Our methods therefore rely on no
unmeasured confounding of $A$ (which holds by design due to the experimental design) and missingness
at random of $Y$ (i.e., $R \ind Y \mid X,A$), which is untestable regardless of whether treatment is randomized. For more details on the trial
and data, we refer to \citet{hammer1996trial} and \citet{wang2018quantile}. \\

Throughout our analysis, we used 5-fold cross-fitting, with all nuisance functions estimated by random forests (via the R package \verb|ranger|  \citep{wright2015ranger}). This includes conditional densities $\eta_a$, which we estimated by regressing a Gaussian kernel weighted outcome on covariates and treatment, on a grid of $y$ values, with bandwidth chosen by Silverman's rule. Alternative approaches could also be used \citep{hansen2004nonparametric, diaz2011super,izbicki2017converting}, potentially at the expense of some extra computational burden.   \\

First we used the density effect methods from Section \ref{sec:denseffest} to check for evidence of an effect of combination therapy on the density of CD4 count.  Specifically, we used the cross-fit version of the estimator in Proposition \ref{prop:l2denseff} to estimate the $L_2^2$ distance between $p_1$ and $p_0$, with asymptotic variance estimated as usual, via the empirical variance of the estimated influence function. To ease interpretability we rescaled $Y$ to be on the unit interval. The estimated $L_2^2$ distance was 0.279 with a 95\% confidence interval of $[0.142, 0.415]$, indicating a statistically significant effect of combination therapy on CD4 count. \\

To more precisely understand how combination therapy impacted the CD4 distribution, we estimated the counterfactual densities using the methods of Sections \ref{sec:densfnest} and \ref{sec:modselest}. Specifically, we used $L_2^2$ projections onto the linear series in Example \ref{ex:series} with the cosine basis \eqref{eq:cosbasis}. We considered a range of models for both densities, including up to 15 basis terms  (more than 15 terms did not improve fit). Figure \ref{fig:tuning} shows estimates of  model fit via the pseudo-$L_2^2$ risk  \eqref{eq:pseudol2}, along with confidence intervals, indicating that four basis terms does best for both counterfactual densities. Figure \ref{fig:results} shows the estimated counterfactual CD4 densities using four basis terms, along with pointwise CIs. Since the densities differ more substantially in the lowest CD4 range (e.g., 0-200), this suggests combination therapy may have increased CD4 count most for the high-risk patients with the lowest counts under control (zidovudine).

\begin{figure}[h!]
\begin{center}
{\includegraphics[width=.95\textwidth]{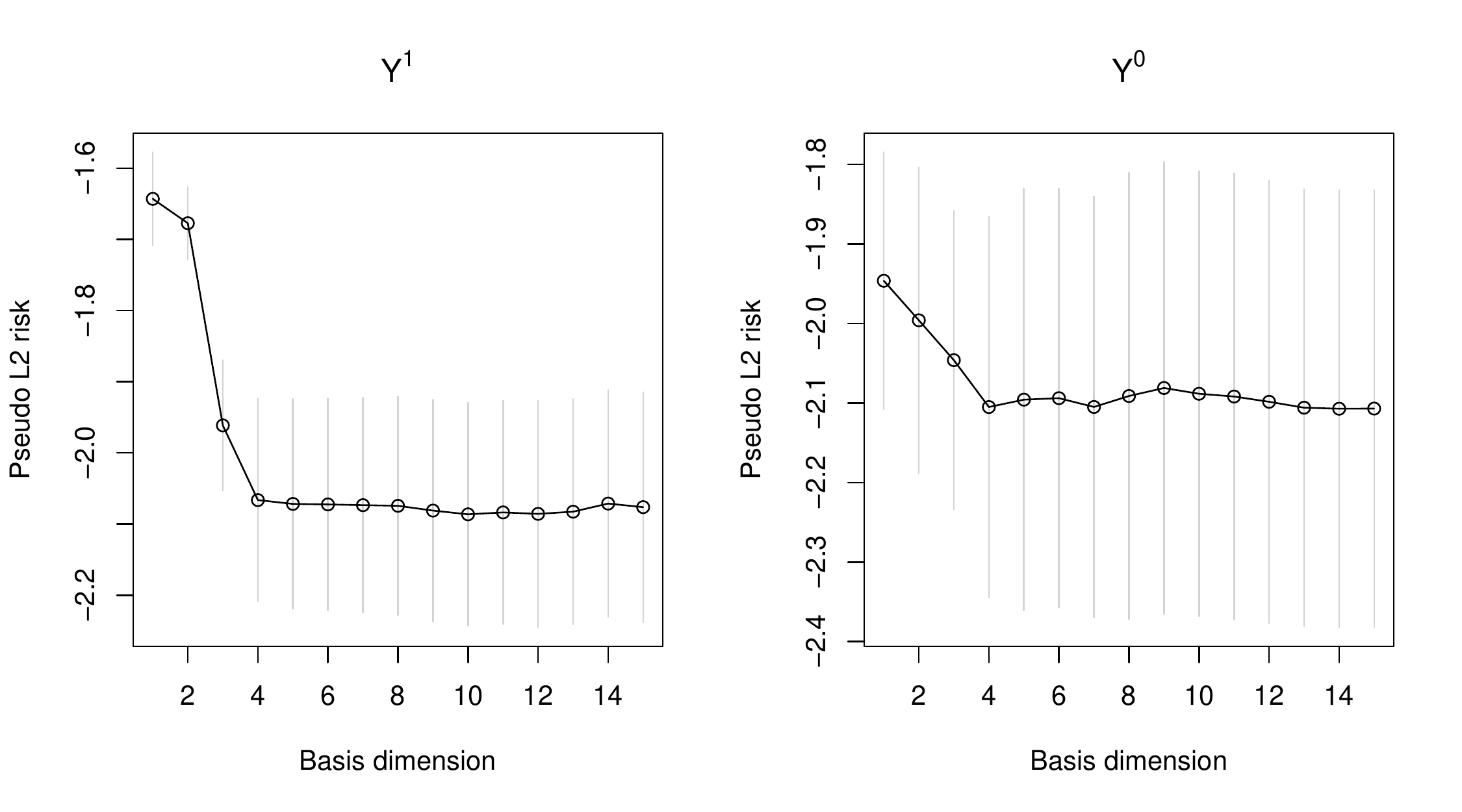}}
\caption{\em Estimates of pseudo-$L_2^2$ risk for models of increasing dimension (using $L_2$ projections onto linear models with a cosine basis), with gray bars denoting confidence intervals.} \label{fig:tuning}
\end{center}
\end{figure}

\begin{figure}[h!]
\begin{center}
{\includegraphics[width=.65\textwidth]{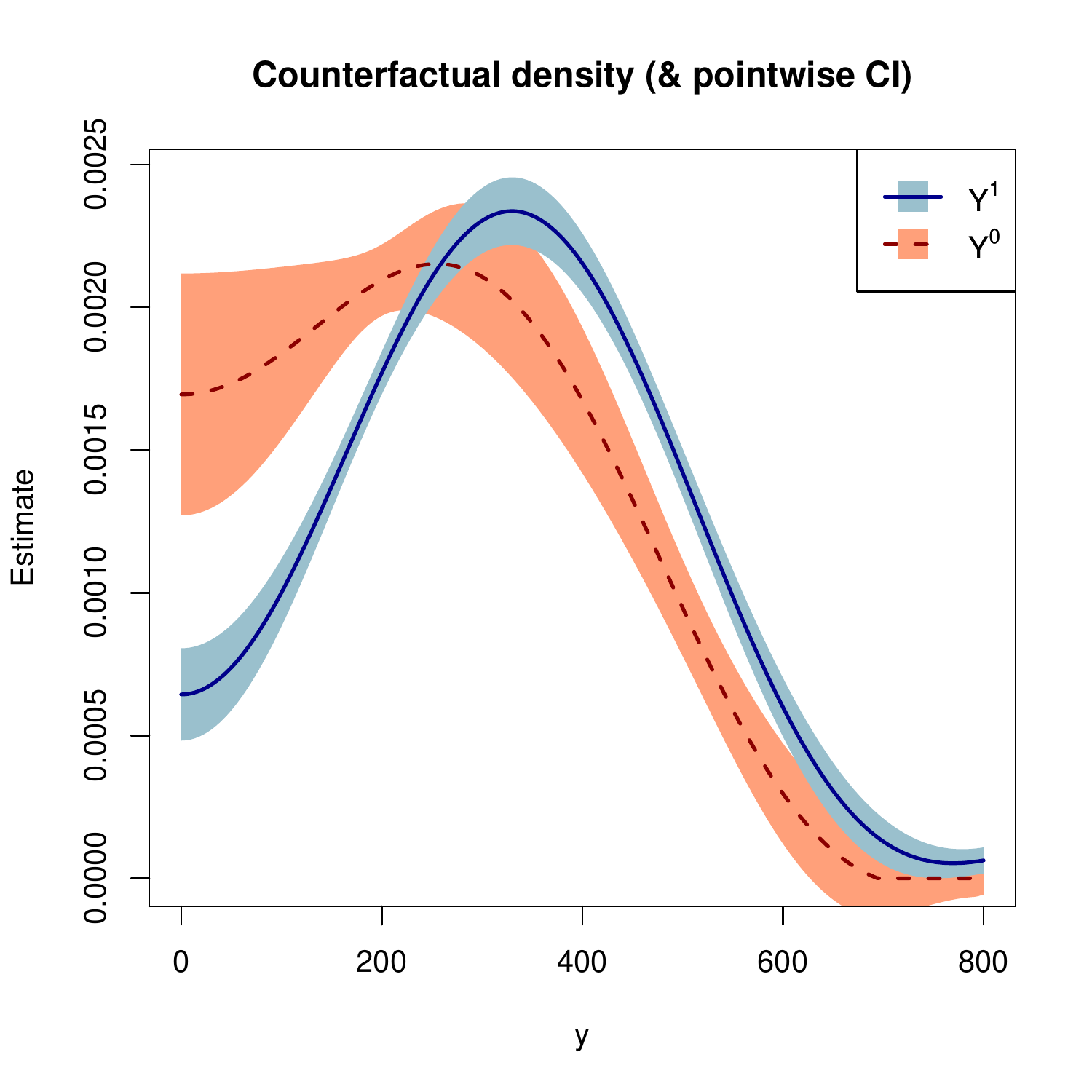}}
\caption{\em Estimated counterfactual CD4 densities for combination therapy versus zidovudine.} \label{fig:results}
\end{center}
\end{figure}

\section{Discussion}

In this paper we proposed methods for estimating counterfactual densities and corresponding distances and other functionals. We gave nonparametric efficiency bounds and flexible optimal estimators for a wide class of models and projection distances, and for new effects that quantify treatment impacts on the density scale.  We also gave methods for data-driven model selection and aggregation in this context, and illustrated the ideas in an application studying effects of antiretroviral therapy on CD4 count. \\ 

There are many interesting avenues for future work. In upcoming companion papers, we consider the nonparametric version of the problem (where the target is the density $p_a$ itself and not a projection) as well as non-discrete treatments (where $A$ is for example a continuous dose). Much more work is needed on the computational side since, outside of $L_2^2$ projections on linear models and KL projections on exponential families, our methods require solving somewhat complicated estimating equations.  Other extensions could involve time-varying treatments, instrumental variables, conditional effects, density-optimal treatment regimes, mediation, sensitivity analysis,  and more. It is also of interest to apply the methods more broadly, to see if they bring any new insights about treatment mechanisms or ways to adapt treatment policies. \\

\section*{Acknowledgements}

Edward Kennedy gratefully acknowledges support from NSF Grant DMS1810979, and  Sivaraman Balakrishnan and Larry Wasserman from NSF Grant DMS1713003. \\

\appendix
 
\section{Appendix: R Code}
\label{sec:code}

\begin{verbatim}
set.seed(100)

# install npcausal package
install.packages("devtools"); library(devtools)
install_github("ehkennedy/npcausal"); library(npcausal)

# load data
library(speff2trial); data(ACTG175); dat <- ACTG175[,c(2:17,19,21,23)]
x <- dat[,!(colnames(dat) %in% c("treat","cd496"))]

# create treatment*missing indicator
a1 <- dat$treat*(!is.na(dat$cd496)); a0 <- (1-dat$treat)*(!is.na(dat$cd496))
a <- a1; a[a0==0 & a1==0] <- -1; y <- dat$cd496; y[is.na(dat$cd496)] <- 0

# estimate pseudo-l2 risk for k=1:15
cv.cdensity(y,a,x, kmax=15, gridlen=50,nsplits=5)

# estimate densities at k=4
res <- cdensity(y,a,x, kmax=4, kforplot=c(4,4), gridlen=50,nsplits=5,ylim=c(0,800))
\end{verbatim}

\pagebreak

\section*{References}
\vspace{-1cm}
\bibliographystyle{abbrvnat}
\bibliography{bibliography}

\pagebreak

\setlength{\parindent}{0cm}

\section{Appendix: Proofs}
\label{sec:proofs}

\subsection{Proof of Corollaries \ref{cor:mb}--\ref{cor:eif_l2kl} and \ref{thm:eifdenseff_kl}}
\label{sec:fexamples}

These corollaries all follow from the distance-specific form of $f$. For reference we list the relevant quantities here. \\

For $L_2^2$ distance we have
\begin{align*}
f(p,q) &= \frac{(p-q)^2}{q} 	& f_1'(p,q) &= 2 \left( \frac{p}{q} - 1 \right) \\
f_2'(p,q) &= 1 - \left( \frac{p}{q} \right)^2	& f_{21}''(p,q) &= - \frac{2p}{q^2} .
\end{align*}

\medskip

For KL divergence we have
\begin{align*}
f(p,q) &= \left(\frac{p}{q} \right) \log \left(\frac{p}{q} \right)				&  f_1'(p,q)  &= \frac{1}{q} \left\{ \log \left(\frac{p}{q} \right) + 1 \right\} \\
f_2'(p,q) &=  -\frac{p}{q^2} \left\{ \log \left(\frac{p}{q} \right) + 1 \right\} 	& f_{21}''(p,q) &= -\frac{1}{q^2} \left\{ \log \left(\frac{p}{q} \right) + 2 \right\} .
\end{align*}

\medskip

For $\chi^2$ divergence we have
\begin{align*}
f(p,q) &= \left(\frac{p}{q} -1 \right)^2 	& f_1'(p,q)   &= \frac{2(p-q)}{q^2} \\
f_2'(p,q) &= -\frac{2p}{q^3}(p-q)		& f_{21}''(p,q) &= \frac{2(q-2p)}{q^3}  .
\end{align*}

\medskip

For Hellinger divergence we have
\begin{align*}
f(p,q) &= \left(\sqrt{\frac{p}{q}}-1\right)^2 	& f_1'(p,q) &= \frac{1}{\sqrt{q}} \left( \frac{1}{\sqrt{q}} - \frac{1}{\sqrt{p}} \right) \\
f_2'(p,q)  &= \frac{ \sqrt{p}}{q^2} \left( \sqrt{q} - \sqrt{p} \right)	& f_{21}''(p,q) &= \frac{\sqrt{q/p} - 2}{2q^2} .
\end{align*}

\medskip

For TV$^*$ divergence we have
\begin{align*}
f(p,q) &= \frac{1}{2q}  \nu_t(p-q) & 
f_1'(p,q)   &=  \frac{1}{2q} \nu_t'(p-q) \\
f_2'(p,q)  &=  \frac{-1}{2q} \left\{ \frac{\nu_t(p-q)}{q} + \nu_t'(p-q) \right\} & 
f_{21}''(p,q) &=  \frac{-1}{2q} \left\{ \frac{\nu_t'(p-q)}{q} + \nu_t''(p-q) \right\}  .
\end{align*}

\bigskip

\subsection{Proof of Lemma \ref{lem:densfunc}}

Here we let $\psi = \psi(\Pb) =\int h(p_a(y)) \ dy$, for
some twice continuously differentiable function $h$. We will show that $\psi$ satisfies the von Mises expansion 
given in Lemma \ref{lem:densfunc}. \\

Let $\overline{p}_a(y) = \int \overline\eta_a(y \mid x) \ d\overline\Pb(x)$ denote the marginal counterfactual density under $\overline\Pb$. Note for the posited influence function given by
\begin{align*}
\varphi(z;\Pb) &= \frac{\one(A=a)}{\pi_a(X)} \left\{ h'(p_a(Y)) - \int h'(p_a(y)) \eta_a(y \mid X) \ dy  \right\} \\
& \hspace{.4in} + \int h'(p_a(y)) \eta_a(y \mid X) \ dy - \int h'(p_a(y)) \eta_a(y \mid x) \ dy \ d\Pb(x) ,
\end{align*}
we have, by iterated expectation, that it has mean under $\Pb$ equal to 
\begin{align*}
 \int \varphi(z;\overline\Pb) \ d\Pb &= \int \bigg[  \frac{\one(A=a)}{\overline\pi_a(X)} \left\{ h'(\overline{p}_a(Y)) - \int h'(\overline{p}_a(y)) \overline\eta_a(y \mid X) \ dy  \right\} \\
& \hspace{.4in} + \int h'(\overline{p}_a(y)) \overline\eta_a(y \mid X) \ dy - \int \int h'(\overline{p}_a(y)) \overline\eta_a(y \mid x) \ dy \ d\overline\Pb(x) \bigg] \ d\Pb \\
&= \int   \frac{\pi_a(x)}{\overline\pi_a(x)}  \int \Big\{ h'(\overline{p}_a(y)) \eta_a(y \mid x)  -  h'(\overline{p}_a(y)) \overline\eta_a(y \mid x) \Big\} \ dy  \ d\Pb(x)  \\
& \hspace{.4in} + \int \int h'(\overline{p}_a(y)) \overline\eta_a(y \mid x) \ dy \ \Big\{ d\Pb(x) - d\overline\Pb(x) \Big\}    \\
&= \int \int h'(\overline{p}_a(y))  \left\{ \frac{\pi_a(x)}{\overline\pi_a(x)} - 1 \right\} \Big\{  \eta_a(y \mid x)  -  \overline\eta_a(y \mid x) \Big\}   \ dy \ d\Pb(x) \\
& \hspace{.4in} + \int h'(\overline{p}_a(y)) \int \Big\{ \eta_a(y \mid x) \ d\Pb(x) - \overline\eta_a(y \mid x) \ d\overline\Pb(x)  \Big\}  \ dy \\
&= \int \int h'(\overline{p}_a(y))  \left\{ \frac{\pi_a(x)}{\overline\pi_a(x)} - 1 \right\} \Big\{  \eta_a(y \mid x)  -  \overline\eta_a(y \mid x) \Big\}   \ dy \ d\Pb(x) \\
& \hspace{.4in} + \int h'(\overline{p}_a(y)) \Big\{ p_a(y)  - \overline{p}_a(y) \Big\}  \ dy . 
\end{align*}
Therefore the second-order remainder term in the von Mises expansion is
\begin{align*}
R_2(\overline\Pb,\Pb) &\equiv \psi(\overline\Pb) - \psi(\Pb) - \int \varphi(z;\overline\Pb) \ d(\overline\Pb-\Pb) = \psi(\overline\Pb) - \psi(\Pb) + \int \varphi(z;\overline\Pb) \ d\Pb \\
&=  \int \int h'(\overline{p}_a(y))  \left\{ \frac{\pi_a(x)}{\overline\pi_a(x)} - 1 \right\} \Big\{  \eta_a(y \mid x)  -  \overline\eta_a(y \mid x) \Big\}   \ dy \ d\Pb(x) \\
& \hspace{.4in} + \int h'(\overline{p}_a(y)) \Big\{ p_a(y)  - \overline{p}_a(y) \Big\}  \ dy +  \int \Big\{ h(\overline{p}_a(y)) -  h(p_a(y))  \Big\} \ dy  \\
&= \int \int h'(\overline{p}_a(y)) \left\{ \frac{\pi_a(x) }{\overline\pi_a(x)} - 1 \right\}  \Big\{ \eta_a(y \mid x) - \overline\eta_a(y \mid x) \Big\} \ dy \ d\Pb(x) \\
& \hspace{.4in} + \frac{1}{2} \int h''(p^*_a(y)) \Big\{ \overline{p}_a(y) - p_a(y) \Big\}^2 \ dy ,
\end{align*}
where the last line follows by a Taylor expansion with remainder of the mean-value form, with  $p^*_a(y)$ lying between $p_a(y)$ and $\overline{p}_a(y)$.

\subsection{Proof of Theorem \ref{lem:eifmb}}

First, for any fixed $\beta$, we have that each element of the $p$-vector
\begin{align*}
m(\beta) = \int \frac{\partial g(y;\beta)}{\partial \beta} 
\left\{ f\Big( p_a(y), g(y;\beta) \Big) + g(y;\beta) f_2'\Big( p_a(y) , g(y;\beta) \Big)   \right\} dy , 
\end{align*}
can be viewed as a density functional $\int h(p_a(y)) \ dy$ for a specific function $h$. In particular, let $g_j'(y;\beta)$ denote the $j^{th}$ element of $\frac{\partial g(y;\beta)}{\partial \beta}$ so that $h(p_a(y))= \{h_1(p_a(y)), ..., h_d(p_a(y))\}^\T$ for 
\begin{equation}
h_j(p_a(y)) =  g_j'(y; \beta) \left\{ f\Big( p_a(y), g(y;\beta) \Big) + g(y;\beta) f_2'\Big( p_a(y) , g(y;\beta) \Big)   \right\} , \label{eq:hdefthm1}
\end{equation}
noting that, for a given $\beta$ value, $g(y;\beta)$ is a known constant not depending on $\Pb$. 

Now we apply  Lemma \ref{lem:densfunc} to each component of $m$. First note that 
\begin{align*}
h_j'(p_a(y)) &=   g_j'(y; \beta) \left\{ f'_1\Big( p_a(y), g(y;\beta) \Big) + g(y;\beta) f''_{21}\Big( p_a(y) , g(y;\beta) \Big)   \right\} , 
\end{align*}
by the chain rule, 
so that
\begin{align*}
\gamma_f(y;\beta) &= \frac{\partial g(y;\beta)}{\partial \beta} \left\{ f'_1\Big( p_a(y), g(y;\beta) \Big) + g(y;\beta) f''_{21}\Big( p_a(y) , g(y;\beta) \Big)   \right\} \\
&= \Big\{ h_1'(p_a(y)), ..., h_p'(p_a(y)) \Big\}^\T .
\end{align*}

Therefore  Lemma \ref{lem:densfunc} implies that $\overline{m}(\beta) = \int h(\overline{p}_a(y)) \ dy$ satisfies the von Mises expansion 
\begin{align}
\overline{m}(\beta) -m(\beta) =\int \varphi_m(z;\overline\Pb) \ d(\overline\Pb-\Pb) + R_2(\overline\Pb,\Pb) ,
\end{align}
where  
\begin{align*}
\varphi_m(Z,\beta;\Pb) &= \frac{\one(A=a)}{\pi_a(X)} \left\{ \gamma_f(Y;\beta) - \int \gamma_f(y;\beta) \eta_a(y \mid X) \ dy \right\} \\
&\hspace{.5in} + \int \gamma_f(y;\beta) \eta_a(y \mid X) \ dy - \int \int \gamma_f(y;\beta) \eta_a(y \mid x) \ dy \ d\Pb(x) , 
\end{align*}
and where the $j^{th}$ component of $R_2(\overline\Pb,\Pb)$ is given by 
\begin{align}
R_{2,j}(\overline\Pb,\Pb) &= \int \int h_j'(\overline{p}_a(y)) \left\{ \frac{\pi_a(x) }{\overline\pi_a(x)} - 1 \right\}  \Big\{ \eta_a(y \mid x) - \overline\eta_a(y \mid x) \Big\} \ dy \ d\Pb(x) \nonumber \\
& \hspace{.4in} + \frac{1}{2} \int h_j''(p^*_a(y)) \Big\{ \overline{p}_a(y) - p_a(y) \Big\}^2 \ dy . \label{eq:mbrem}
\end{align}

Now we give a lemma showing why finding a von Mises expansion like the above, with second-order remainder, is equivalent to finding the efficient influence function in a nonparametric model. This will prove $\varphi_m$ is the efficient influence function for $m(\beta)$, and will also be useful for later results.

\begin{lemma} \label{lem:exptoeif}
Let $\psi: \mathcal{P} \rightarrow \R$ denote some real-valued functional on a nonparametric model, so the set of distributions $\mathcal{P}$ does not constrain the tangent space. Assume the functional satisfies
$$ \psi(\overline\Pb) - \psi(\Pb) = \int \varphi(z;\overline\Pb) \ (d\overline\Pb-d\Pb) + R_2(\overline\Pb,\Pb)$$
for some mean-zero and finite variance function $\varphi(z;\Pb)$. Then $\varphi$ is the efficient influence influence function if $ \frac{d}{d\epsilon} R_2(\Pb,\Pb_\epsilon) |_{\epsilon=0} =0$ for any smooth parametric submodel.
\end{lemma}

\begin{proof}
Recall from  \citet{bickel1993efficient} and \citet{van2002semiparametric} that the efficient influence function is the mean-zero function whose variance equals the nonparametric efficiency bound, and  is given by the unique function $\phi$ that is a valid submodel score (or limit of such scores) satisfying pathwise differentiability, i.e.,
\begin{equation} \label{eq:pathwise}
 \frac{d}{d\epsilon} \psi(\Pb_\epsilon) \Bigm|_{\epsilon=0} = \int  \phi(z;\Pb) \left( \frac{d}{d\epsilon} \log d\Pb_\epsilon \right) \Bigm|_{\epsilon=0} d\Pb(z)
 \end{equation}
for $\Pb_\epsilon$ any smooth parametric submodel. (e.g., differentiable in quadratic mean) In a nonparametric model only one such function $\phi$ satisfies the above. We will show that the above is satisfied by the function $\varphi$ in the statement of the lemma. 

First note that the assumed expansion implies
$$ \psi(\Pb) - \psi(\Pb_\epsilon) = - \int \varphi(z;\Pb) \ d\Pb_\epsilon + R_2(\Pb,\Pb_\epsilon) $$
for any submodel $\Pb_\epsilon$. Differentiating with respect to $\epsilon$ gives
\begin{align*}
 \frac{d}{d\epsilon} \psi(\Pb_\epsilon) &=  \frac{d}{d\epsilon} \int \varphi(z;\Pb) \ d\Pb_\epsilon + \frac{d}{d\epsilon} R_2(\Pb,\Pb_\epsilon) \\
 &= \int \varphi(z;\Pb) \left( \frac{d}{d\epsilon}  \log d\Pb_\epsilon \right) \ d\Pb_\epsilon + \frac{d}{d\epsilon} R_2(\Pb,\Pb_\epsilon)  , 
\end{align*}
where the second line follows from the dominated convergence theorem and uses the fact that $\frac{d}{d\epsilon}  \log d\Pb_\epsilon=\frac{d}{d\epsilon} d\Pb_\epsilon / d\Pb_\epsilon$. Therefore evaluating at $\epsilon=0$ we have
$$  \frac{d}{d\epsilon} \psi(\Pb_\epsilon) \Bigm|_{\epsilon=0} = \int \varphi(z;\Pb) \left( \frac{d}{d\epsilon}  \log d\Pb_\epsilon \right) \Bigm|_{\epsilon=0} \ d\Pb + \frac{d}{d\epsilon} R_2(\Pb,\Pb_\epsilon)\Bigm|_{\epsilon=0}  , $$
which yields the desired pathwise differentiability by the fact that $ \frac{d}{d\epsilon} R_2(\Pb,\Pb_\epsilon) |_{\epsilon=0} =0$.
\end{proof}

Now we can immediately apply Lemma \ref{lem:exptoeif}, noting that 
$$ \frac{d}{d\epsilon} R_2(\Pb, \Pb_\epsilon) \Bigm|_{\epsilon=0} = 0 $$
by virtue of the fact that the remainder $R_{2,j}(\Pb,\Pb_\epsilon)$ in \eqref{eq:mbrem} consists of only second-order products of errors between $\Pb$ and $\Pb_\epsilon$. This follows since applying the product rule yields a sum of two terms, each of which is a product of a derivative term (which may not be zero at $\epsilon=0$) and an error term involving differences of components of $\Pb_\epsilon$ and $\Pb$ (which will be zero at $\epsilon=0$). Therefore $\varphi_m$ is the efficient influence function for the parameter $m(\beta)$. The efficient influence functions for $\beta_0$ and $g(y;\beta_0)$ follow similarly, via the chain rule. 

\subsection{Proof of Theorem \ref{thm:eifdenseff}}

From Lemmas \ref{lem:densfunc} and \ref{lem:exptoeif}, the efficient influence function of $\psi_f=\int f\left( {p_1(y)} , {p_0(y)} \right) p_0(y) \ dy$ if $p_0(y)$ were known would be 
\begin{align*}
\varphi_1(z;\Pb) &= \frac{\one(A=1)}{\pi(1 \mid X)} \left\{ h_1'(Y) - \int h_1'(y) \eta_1(y \mid X) \ dy  \right\} \\
& \hspace{.4in} + \int h_1'(y) \eta_1(y \mid X) \ dy - \int h_1'(y) \eta_1(y \mid x) \ dy \ d\Pb(x)
\end{align*}
where
$$ h_1'(y) = h_1'(y;p_0,p_1) = p_0(y) f'_1(p_1(y),p_0(y)) . $$
Similarly, if $p_1(y)$ were known, the efficient influence function of $\psi_f$ would be
\begin{align*}
\varphi_0(z;\Pb) &= \frac{\one(A=0)}{\pi(0 \mid X)} \left\{ h_0'(y) - \int h_0'(y) \eta_0(y \mid X) \ dy  \right\} \\
& \hspace{.4in} + \int h_0'(y) \eta_0(y \mid X) \ dy - \int h_0'(y) \eta_0(y \mid x) \ dy \ d\Pb(x) ,
\end{align*}
where
$$ h_0'(y) = h_0'(y;p_0,p_1) = f(p_1(y),p_0(y)) + p_0(y) f'_2(p_1(y),p_0(y))   . $$
The result then follows from the fact that the influence function when $p_1$ and $p_0$ are  both unknown is the sum of the two influence functions when $p_1$ and $p_0$ are known, separately. 

\subsection{Proof of Claim in Remark \ref{rem:initdens}}
\label{app:initdens}

Here we show why rates for estimating ${p}_a(y)$ with the plug-in estimator $\widehat{p}_a(y) = \Pn\{ \widehat\eta_a(y \mid X)\}$ will not be slower than those for estimating $\widehat\eta_a(y \mid x)$, by bounding the mean squared error of the former in terms of the latter. To this end we denote the pointwise bias and variance of $\widehat\eta_a$ as $\E \{\widehat\eta_a(y \mid x) \}- \eta_a(y \mid x)= b(y \mid x)$ and $\var\{ \widehat\eta_a(y \mid x)\} = v(y \mid x)$, respectively. First note for the bias that
\begin{align*}
\E\{ \widehat{p}_a(y)\} - p_a(y) &=  \E\Big[ \Pn\{ \widehat\eta_a(y \mid X)\} \Big] - \int \eta_a(y \mid x) \ d\Pb(x) \\
&= \E \int \widehat\eta_a(y \mid x) \ d\Pb(x) - \int \eta_a(y \mid x) \ d\Pb(x) =  \int b(y \mid x) \ d\Pb(x) ,
\end{align*}
where in the second line we used iterated expectation, conditioning on the training sample $D^n$ used to construct $\widehat\eta_a$. For the variance we similarly have 
\begin{align*}
\var\{ \widehat{p}_a(y) \} &= \var\left( \E\Big[ \Pn\{ \widehat\eta_a(y \mid X)\} \mid D^n \Big] \right) + \E\left( \var\Big[ \Pn\{ \widehat\eta_a(y \mid X)\} \mid D^n \Big] \right)  \\
&= \var \int \widehat\eta_a(y \mid x) \ d\Pb(x) +  \frac{1}{n} \E \left[ \var \Big\{ \widehat\eta_a(y \mid X) \mid D^n \Big\} \right] . 
\end{align*}
For the first term above, by Cauchy-Schwarz we have
\begin{align*}
\var \int \widehat\eta_a(y \mid x) \ d\Pb(x) &= \E\left\{ \int \Big[ \widehat\eta_a(y \mid x) - \E\{\widehat\eta_a(y \mid x) \} \Big] \ d\Pb(x) \right\}^2 \\
&\leq \E \int \Big[ \widehat\eta_a(y \mid x) - \E\{\widehat\eta_a(y \mid x) \} \Big]^2 \ d\Pb(x) = \int v(y \mid x) \ d\Pb(x)
\end{align*}
And for  the second term note that
\begin{align*}
\E \left[ \var \Big\{ \widehat\eta_a(y \mid X) \mid D^n \Big\} \right] &\leq \E   \int  \widehat\eta_a(y \mid x)^2 \ d\Pb(x)  = \int \left( v(y \mid x) + \Big[ \E \{ \widehat\eta_a(y \mid x) \} \Big]^2 \right) \ d\Pb(x) \\
&= \int \left( v(y \mid x)  + \left[ \E\Big\{ \widehat\eta_a(y \mid x) - \eta_a(y \mid x) + \eta_a(y \mid x) \Big\} \right]^2 \right) \ d\Pb(x) \\
&\leq \int  \Big\{ v(y \mid x) + 2b(y \mid x)^2 + 2\eta_a(y \mid x)^2 \Big\} \ d\Pb(x) . 
\end{align*}
Therefore as long as $\int \eta_a(y \mid x)^2 d\Pb(x) \leq C$, we have
\begin{align*}
 \E\left[\Big\{ \widehat{p}_a(y) - p_a(y) \Big\}^2\right] \leq \left(1 + \frac{2}{n} \right) \int \E\left[ \Big\{ \widehat\eta_a(y \mid x) - \eta_a(y \mid x) \Big\}^2 \right] d\Pb(x) + \frac{2C}{n} 
\end{align*}

\subsection{Proof of Theorem \ref{thm:densfnrate}}

First we present a master lemma giving the rate of convergence of the solution to a sample-split estimating equation. The logic parallels that of Theorem 5.31 of \citet{van2000asymptotic}.

\begin{lemma} \label{lem:esteq}
Let $\varphi(z;\theta,\eta)$ denote a vector estimating function for target parameter $\theta \in \R^p$ and nuisance functions $\eta \in H$ for some function space $H$.  Suppose the true values $(\theta_0,\eta_0)$ satisfy $\Pb\{ \varphi(Z;\theta_0,\eta_0)\}=0$, and define the estimator $\widehat\theta$ as an approximate solution to the estimating equation satisfying
$$ \Pn \{\varphi(Z;\widehat\theta,\widehat\eta) \}= o_\Pb(1/\sqrt{n}) $$
where $\widehat\eta$ is estimated on a separate independent sample. Assume:
\begin{enumerate}
\item The function class $\{ \varphi(z;\theta,\eta) : \theta \in \R^p\}$ is Donsker in $\theta$ for any fixed $\eta$.
\item The estimators are consistent, i.e., $\widehat\theta-\theta_0 = o_\Pb(1)$ and $\| \widehat\eta - \widehat\eta_0 \| = o_\Pb(1)$.
\item The map $\theta \mapsto \Pb\{ \varphi(Z;\theta,\eta)\}$ is differentiable at $\theta_0$ uniformly in $\eta$, with nonsingular derivative matrix $\frac{\partial}{\partial\theta} \Pb\{\varphi(Z;\theta,\eta)\} |_{\theta=\theta_0} = V(\theta_0,\eta)$, where $V(\theta_0,\widehat\eta) \inprob V(\theta_0,\eta_0)$.
\end{enumerate}
Then 
$$ \widehat\theta - \theta_0 = -V(\theta_0,\eta_0)^{-1} (\Pn-\Pb)  \{ \varphi(Z;\theta_0,\eta_0) \} + O_\Pb\Big(R_n\Big) + o_\Pb\left( \frac{1}{\sqrt{n}} \right) $$
for $R_n = \Pb\{ \varphi(Z;\theta_0,\widehat\eta) - \varphi(Z;\theta_0,\eta_0)\} $.
\end{lemma}

\begin{proof}
First note that, since $(\widehat\theta,\widehat\eta)$ and $(\theta,\eta)$ are approximate and exact solutions of the empirical and population moment conditions, respectively, we have
\begin{align}
o_\Pb(1/\sqrt{n}) &= \Pn \{ \varphi(Z;\widehat\theta,\widehat\eta) \} - \Pb\{ \varphi(Z;\theta_0,\eta_0)\} \nonumber \\
&= (\Pn-\Pb) \{ \varphi(Z;\theta_0,\eta_0) \}   + (\Pn-\Pb) \{ \varphi(Z;\widehat\theta,\widehat\eta) - \varphi(Z;\theta_0,\widehat\eta) \} \label{eq:eet1} \\
& \hspace{.5in} + (\Pn-\Pb) \{ \varphi(Z;\theta_0,\widehat\eta) - \varphi(Z;\theta_0,\eta_0) \}  \label{eq:eet2} \\
& \hspace{.5in} + \Pb\{ \varphi(Z;\widehat\theta,\widehat\eta) - \varphi(Z;\theta_0,\widehat\eta)\} + \Pb\{ \varphi(Z;\theta_0,\widehat\eta) - \varphi(Z;\theta_0,\eta_0)\}  \label{eq:eet3}
\end{align}
where the second equality follows by simply adding and subtracting terms. The first term in \eqref{eq:eet1}  is a simple sample average of a fixed function and so will be asymptotically Gaussian by the central limit theorem. The second term in \eqref{eq:eet1} and the term in \eqref{eq:eet2}  are empirical process terms. The first term in \eqref{eq:eet3} will be linearized in $(\widehat\theta - \theta_0)$, while the second term in \eqref{eq:eet3} captures the effect of the nuisance estimation error. We will tackle each of these in turn. 

Under the Donsker and consistency conditions for $\widehat\theta$ in Assumptions 1 and 2, the second term in \eqref{eq:eet1} is $o_\Pb(1/\sqrt{n})$ by Lemma 19.24 of \citet{van2000asymptotic}. Under the consistency of $\widehat\eta$ in Assumption 2 and the sample splitting, the term in \eqref{eq:eet2} is $o_\Pb(1/\sqrt{n})$ by Lemma 2 of \citet{kennedy2020sharp}. 

By the differentiability of the map $\theta \mapsto \Pb\{ \varphi(Z;\theta,\eta)\}$ in Assumption 3, the first term in \eqref{eq:eet3} can be expressed as
\begin{align*}
\Pb\{ \varphi(Z;\widehat\theta,\widehat\eta) - \varphi(Z;\theta_0,\widehat\eta)\} &= V(\theta_0,\widehat\eta) (\widehat\theta - \theta_0) + o_\Pb(\|\widehat\theta-\theta_0\|) \\
&=  V(\theta_0,\eta_0) (\widehat\theta - \theta_0) + o_\Pb(\|\widehat\theta-\theta_0\|)
\end{align*}
where the last line follows by the consistency of $V(\theta_0,\widehat\eta)$ in Assumption 3.

Therefore we have
\begin{align*}
o_\Pb(1/\sqrt{n}) &= (\Pn-\Pb) \{ \varphi(Z;\theta_0,\eta_0) \}  +  V(\theta_0,\eta_0) (\widehat\theta - \theta_0) + R_n + o_\Pb(\|\widehat\theta-\theta_0\|)
\end{align*}
where we let $R_n=\Pb\{ \varphi(Z;\theta_0,\widehat\eta) - \varphi(Z;\theta_0,\eta_0)\} $ denote the second term in \eqref{eq:eet3}, 
or equivalently
$$ \widehat\theta - \theta_0  = -V(\theta_0,\eta_0)^{-1} (\Pn-\Pb)  \{ \varphi(Z;\theta_0,\eta_0) \} + O_\Pb(R_n) + o_\Pb(\| \widehat\theta - \theta_0\| ) + o_\Pb(1/\sqrt{n}) 
$$
by the nonsingularity of the derivative matrix in Assumption 3. This implies
$$ \| \widehat\theta - \theta_0 \| (1 + o_\Pb(1)) = O_\Pb( 1/\sqrt{n} + R_n ) $$
so that $\| \widehat\theta - \theta_0 \| = O_\Pb(1/\sqrt{n} + R_n)$, which gives the result after noting that $o_\Pb(O_\Pb(1/\sqrt{n} + R_n))=o_\Pb(1/\sqrt{n} + R_n)$ and that  $O_\Pb(R_n) + o_\Pb(R_n) = O_\Pb(R_n)$.
\end{proof}

\bigskip

Now we can apply Lemma \ref{lem:esteq} to prove Theorem  \ref{thm:densfnrate}. First note that by definition the estimator satisfies
$$ \Pn\{  \phi(Z;\widehat\beta,\widehat\eta) \} = o_\Pb(1/\sqrt{n}) $$
for $\phi(Z;\beta,\eta) = m(\beta;\eta) + \varphi(Z;\beta,\eta)$, and the true values $(\beta_0,\eta_0)$ satisfy $\Pb\{ \phi(Z;\beta_0,\eta_0)\}=0$, again by definition.

Conditions 1--3 of Lemma \ref{lem:esteq} hold by Assumptions 1--4 of Theorem \ref{thm:densfnrate}, so the result follows by virtue of the fact that 
\begin{align*}
R_n &\equiv \Pb\{ \phi(Z;\beta_0,\widehat\eta) - \phi(Z;\beta_0,\eta_0)\} \\
&= m(\beta_0;\widehat\eta) + \Pb\{ \varphi_m(Z;\beta_0,\widehat\eta) \} \\
&= \Pb \int h'(\widehat{p}_a(y)) \left\{ \frac{\pi_a(X) }{\widehat\pi_a(X)} - 1 \right\}  \Big\{ \eta_a(y \mid X) - \widehat\eta_a(y \mid X) \Big\} \ dy \nonumber \\
& \hspace{.4in} + \frac{1}{2} \int h''(\widehat{p}^*_a(y)) \Big\{ \widehat{p}_a(y) - p_a(y) \Big\}^2 \ dy \\
&\lesssim \| \widehat\pi - \pi \| \| \widehat\eta_a - \eta_a \| + \delta \| \widehat{p}_a - p_a \|^2 
\end{align*}
where the second to last line follows from the result given in equation \eqref{eq:mbrem} of the proof of  Theorem \ref{lem:eifmb}, for the vectors $h$ and $h'$  as defined in \eqref{eq:hdefthm1}, and the last line follows by the Cauchy-Schwarz inequality and boundedness assumptions on $h'$, $h''$, and $1/\widehat\pi$.

\end{document}